\documentclass[sigconf,nonacm]{acmart}

\renewcommand\footnotetextcopyrightpermission[1]{}
\usepackage{amsmath}
\usepackage{graphicx}
\usepackage{tabularx}
\usepackage{booktabs}
\usepackage{balance}
\usepackage{amsthm}
\usepackage[capitalize,noabbrev]{cleveref}
\usepackage{relsize}
\usepackage{dblfloatfix}
\usepackage{tikz}
\usepackage{pgfplots}

\newtheorem{theorem}{Theorem}
\newtheorem{definition}{Definition}
\newtheorem{lemma}{Lemma}
\newtheorem{proposition}{Proposition}
\newtheorem{corollary}{Corollary}

\begin{document}

\title{Thermodynamic Human-Computer Interaction}

\author{Uzafir Ahmad Rafaq}
\affiliation{
  \institution{Heriot-Watt University}
  \country{U.A.E}}
\email{uzafir525@gmail.com}

\author{Muaz Hassan}
\affiliation{
  \institution{Independent Researcher}
  \country{Ireland}
}
\email{syedmuazhassan@gmail.com}

\author{Ali Muzaffar}
\affiliation{
  \institution{Heriot-Watt University}
  \country{U.A.E}}
\email{a.muzaffar@hw.ac.uk}

\begin{teaserfigure}
    \centering
    \includegraphics[width=0.9\textwidth]{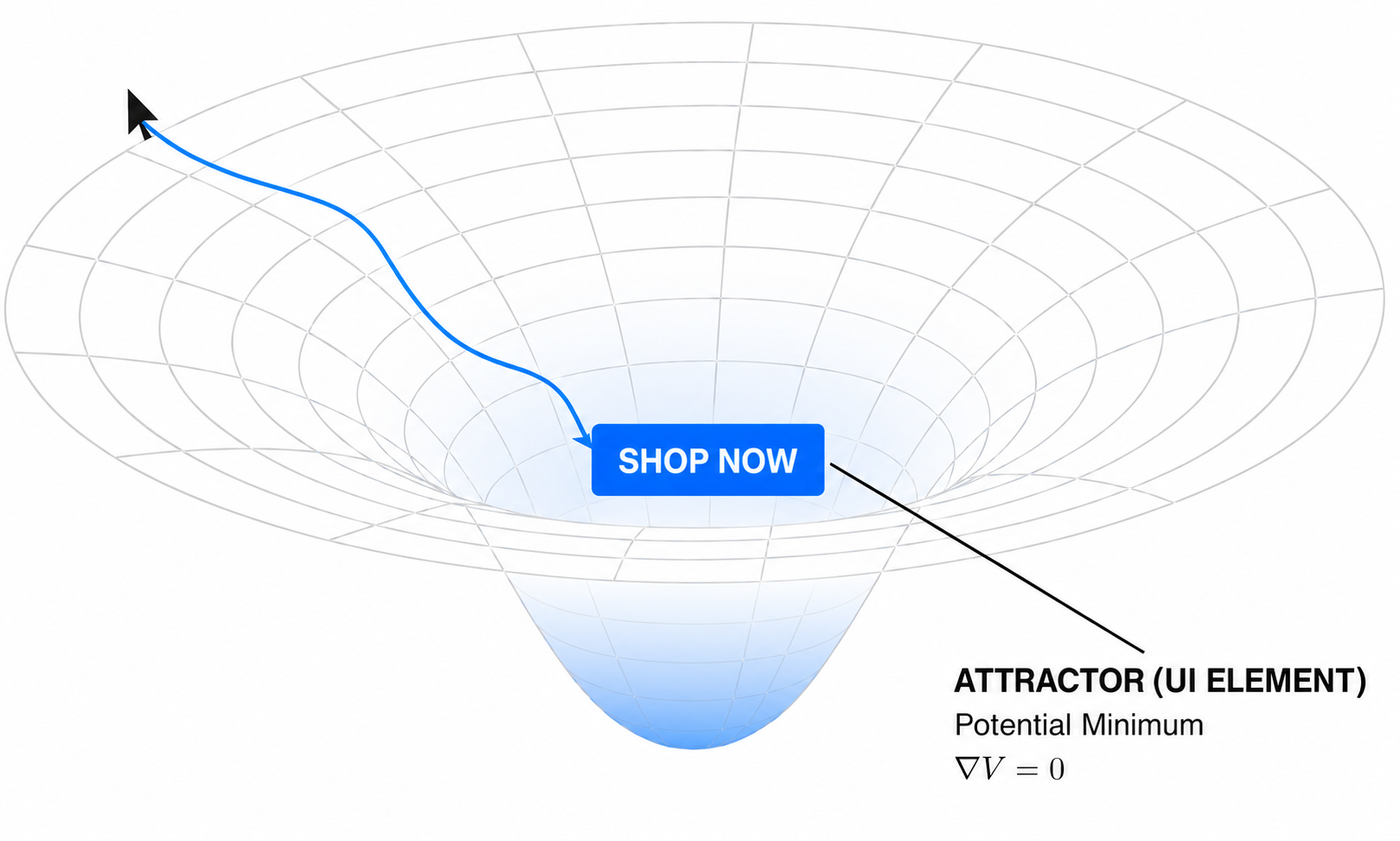}
    \caption{A digital target acts as a harmonic potential well. The cursor follows a noisy path towards the potential minimum ($\nabla V = 0$), while the color, text, and border radius of the button influence the attractive force.}
    \label{fig:teaser_potential_well}
\end{teaserfigure}

\begin{abstract}
Target acquisition is often modeled separately for desktop, mobile, and other interaction modalities. We present Thermodynamic HCI, a framework that splits interaction into thermal equilibrium and non-equilibrium regimes. The theory generalizes across interaction modalities by representing agent-target interaction using kinetic and potential energies. We derive the movement time of Fitts' law and the speed–accuracy tradeoff observed in Schmidt's law from the principles of thermal physics. Furthermore, we develop theorems that describe how target properties, such as the color of a button, affect user accuracy. The target acquisition model, derived from the theory, when evaluated on desktop and mobile website prefetching experiments, achieved an accuracy of 98\% for both cursor and touchscreen based interaction. For every clicked link, it produced a fetch:click ratio of 1.37 for desktop and 1.75 for mobile. \footnote{The target acquisition model is open source and available as the \texttt{intent-link} npmjs package at \href{https://www.npmjs.com/package/intent-link}{npmjs.com/package/intent-link}.}
\end{abstract}



\maketitle

\section{Introduction}
Accurate prediction of a user's intended target in an information system offers substantial benefits. Examples include making target icons sticky \cite{wordenMakingComputersEasier1997}, enlarging the target \cite{mcguffinAcquisitionExpandingTargets2002,mcguffinFittsLawExpanding2005}, or highlighting a likely target at an appropriate moment \cite{yuOptimizingTimingIntelligent2022}. 

Different interaction modalities require different approaches to predicting the user's target. VR systems, mobile devices, and desktop cursors each have their own methods for predicting the target. Most of these methods are designed and evaluated for a specific interaction modality: touch prediction models developed for phones do not transfer directly to desktops, while models based on cursor motion require adaptation for mobile devices. Various modeling options have previously been adopted. 

Early methods relying on kinematic and spatial heuristics \cite{lankEndpointPredictionUsing2007,asanoPredictiveInteractionUsing2005,ahmadYouNotHave2016,murataImprovementPointingTime1998} are computationally lightweight, but represent variability and noise in human motor behavior through limited, task specific assumptions. Models based on learning from data, including neural networks \cite{lePredicTouchSystemReduce2017,biswasIntentRecognitionUsing2013a}, can account for some of this variability by learning patterns from observed behavior. Reinforcement learning has also been used to decide whether a system should select a predicted target directly or present the user with alternatives \cite{liSelectSuggestReinforcement2022}. However, learned models require datasets, and their performance may not generalize across devices without additional data or retraining. In addition, their decisions are often less interpretable than those of kinematic models. The data requirements of neural networks have been partially addressed using biomechanical simulations \cite{moonRealtime3DTarget2024a}, which utilize a virtual human to perform tasks and gather data. However, biomechanical simulation approaches can require significant computational resources and task specific simulator design. Conversely, Optimal Control models \cite{ziebartProbabilisticPointingTarget2012,fischerOptimalFeedbackControl2022,klarSimulatingInteractionMovements2023,martinIntermittentControlModel2021} ground prediction in movement objectives and system dynamics. Although some support real time prediction, richer formulations can demand computationally intensive parameter fitting and numerical optimization, complicating deployment under tight computational budgets. 

Therefore, there remains a need for a unified target prediction model that combines the physical grounding and interpretability of analytical movement models with the computational efficiency of kinematic heuristics. To address the need, this study introduces a theory that supports target prediction across dimensions, for example, both 1D mobile scrolling and 2D cursor interaction and is lightweight and accurate enough to be used in production systems.

Human behavior is often described as favoring actions that require less effort \cite{zipfHumanBehaviorPrinciple2016}. This principle motivates a representation in which interaction states can be described using energy from physical quantities such as position and velocity. Information theory has strongly influenced mathematical models in human-computer interaction. Fitts framed aimed movement as a problem of information transmission and adapted concepts from Shannon's Information theory to relate movement time to target distance and width \cite{fittsInformationCapacityHuman1954,shannonMathematicalTheoryCommunication1948b}. Shannon entropy quantifies uncertainty over possible messages. Jaynes later showed that maximizing the same entropy produces the Gibbs distribution \cite{jaynesInformationTheoryStatistical1957}. This establishes a mathematical connection between information theory and equilibrium statistical mechanics. Subsequently, this relationship provides the motivation behind the theory proposed in this study to understand human-computer interaction through a thermodynamic lens. 

The proposed framework considers the entire interaction interval to be in regimes of thermal equilibrium and non-equilibrium. We assign a kinetic energy to the agent and target from their movement velocities. We also define a potential energy field based on the distance between the agent and the target. We model the stochastic movement of the agent towards the target by utilizing differential equations. We further demonstrate that the movement time in Fitts' Law emerges as a special case for an approach trajectory that resembles that of an underdamped harmonic oscillator.

From the theory, we create a target acquisition model that has a target prediction accuracy of 98\% and works in constant $\mathcal{O}(1)$ time. The model requires no learned training stage and is computationally efficient enough to run on resource constrained devices. To evaluate this model, we chose the application of target acquisition to be the prefetching of website links. For every clicked link, the model achieved a fetch:click ratio of 1.37 for cursor based interaction and 1.75 for touchscreen based interaction. 

We further develop the theory and introduce theorems that suggest each property of a target independently influences the endpoint inaccuracy of the user. For example, for a UI button on a website, its color, border radius and the clarity of its label may all have an effect on the accuracy of the user attempting to click it. We evaluate the effects of label clarity and color contrast on the accuracy of the user in our evaluation section. Additionally, a theorem is derived that describes the speed-accuracy tradeoff observed in HCI \cite{schmidtMotoroutputVariabilityTheory1979, zhaiSpeedAccuracyTradeoff2004, harrisSignaldependentNoiseDetermines1998, plamondonSpeedAccuracyTradeoffs1997}. Furthermore, we develop another theorem to distinguish between the hypothesized equilibrium and non-equilibrium states of the system.

\subsection{Aim, Objectives and Contributions}
\subsubsection{Aim of the Study}
The aim of this study is to develop and evaluate an interpretable and computationally lightweight theory of user interaction across modalities. The framework connects target prediction, movement time, endpoint accuracy, and interface properties within a common thermodynamic representation. 

\subsubsection{Objectives of the Study}
The objectives of this study are to:
\begin{enumerate}
\item Represent an interaction between an agent and a target using a thermodynamic analogy.

\item Formalize a probabilistic model to predict the agent's intended target.

\item Use differential equations to describe movement towards a target and determine how a Fitts shaped movement time can emerge from a damped approach.

\item Derive testable relationships between movement speed, endpoint accuracy, and target properties such as label clarity and color contrast.

\item Implement and evaluate the target acquisition model on a working web prefetching system for one dimensional mobile scrolling and two dimensional cursor movement.

\item Evaluate in a user study how the clarity of a button label and its color may impact the accuracy of the user.

\item Identify the limits of the target acquisition model and define a diagnostic for detecting when its assumptions are not satisfied.
\end{enumerate}

\subsubsection{Contributions of This Study}
This study makes the following contributions:
\begin{enumerate}
\item It provides a framework where candidate interactions are represented by potential and kinetic energy. Langevin dynamics describe the system evolution, and the Gibbs distribution describes the stable probability distribution approached by the system.

\item It introduces an interpretable, efficient target prediction model that requires no training data and scales across interaction modalities. 

\item It derives the movement time used in Fitts' law and the speed accuracy tradeoff observed in Schmidt's law.

\item It introduces a parameterized potential field in which interface properties, such as the color of a button, can change the strength of attraction towards an agent and affect the accuracy of the user.

\item It proposes a diagnostic for distinguishing the equilibrium and non-equilibrium assumptions of the framework. 

\item It provides a functional implementation of the theory by providing a highly efficient prefetching library for mobile and desktop websites through an npmjs package. This presents a practical use case for the industry.
\end{enumerate}

\section{Related Work}
Target prediction can be approached from different perspectives, such as with neural networks or control theory. The applications of successfully predicting a target are numerous and have a significant impact on user experience. 

\subsection{Applications of Target Prediction} 
Once the user's intent to interact with a specific target has been successfully inferred, numerous practical applications arise, such as
\begin{enumerate}
    \item Move candidate targets closer to the pointing location. \cite{PDFDragandPopDragandPick}
    
    \item Restricting pointing to selectable targets. \cite{guiardObjectPointingComplement2004}
    
    \item Sticky Icons. \cite{wordenMakingComputersEasier1997}

    \item Enlarge the candidate targets. \cite{mcguffinFittsLawExpanding2005, mcguffinAcquisitionExpandingTargets2002}

    \item Expand the cursor's interactive area. \cite{chapuisDynaSpotSpeeddependentArea2009, grossmanBubbleCursorEnhancing2005, mottBeatingBubbleUsing2014}

    \item Prefetching website links. \cite{nigamAnalysisMarkovModel2010a}
\end{enumerate}
The application of target prediction in this study is chosen as prefetching. The prefetching process involves a system that calls an API to fetch the contents of a page on a website before the user clicks on the link. This enhances the user experience by providing zero loading and instant navigation. This application is chosen because of its high impact on economics and productivity. Miller \cite{millerResponseTimeMancomputer1968} showed that a computer must respond in under 100 milliseconds to keep the user focused on their current task. If a system waits for a user to initiate a request before starting its computation, the system will feel slow. Small delays cause users to leave a website and decrease their overall engagement. Brutlag showed that adding just a few hundred milliseconds of delay significantly reduces user interactivity \cite{brutlagSpeedMattersGoogle}. Similarly, Kohavi demonstrated that minor increases in latency cause measurable financial loss and user frustration in commercial systems \cite{kohaviOnlineExperimentsLessons2007}. In modern web development, prefetching is often implemented using JavaScript libraries. For example, Quicklink \cite{GoogleChromeLabsQuicklink2026} is a tool that automatically prefetches the URLs of links that are currently visible on the user's screen. ForesightJS \cite{spaansSpaansbaForesightJS2026} takes a different approach by focusing on the trajectory of the cursor to predict which link is about to be clicked. 

\subsection{Spatial \& Kinematic Heuristics} 
Spatial heuristics involve tracking the position of an agent, such as detecting the pointing finger in mid-air \cite{ahmadYouNotHave2016} using motion sensors and applying a probabilistic nearest neighbor algorithm to detect the target which will be selected. However, this inevitably results in one of the selectable targets being ranked higher constantly, despite the user having no intention of interacting with it, and this approach does not incorporate kinematic data. Another spatial heuristic approach is to measure the angle between the velocity vector of the agent and the position vector of the target \cite{murataImprovementPointingTime1998}. Although angle based heuristics successfully incorporate kinematic data, they suffer from inaccuracies when attempting to make predictions in cluttered interfaces, where a lot of targets are in close proximity. Further work on kinematic and spatial heuristics involves tracking the peak movement velocity of the cursor and its direction \cite{asanoPredictiveInteractionUsing2005, lankEndpointPredictionUsing2007, fuYourMouseReveals2017a}. ForesightJS \cite{spaansSpaansbaForesightJS2026} is a JavaScript library which is built upon these kinematic and spatial heuristics. It involves projecting the cursor velocity and direction forward in time to predict the endpoint of interaction. All these kinematic and spatial heuristics help determine the endpoint of interaction but do not provide enough information on the click event itself. The Intermittent Click Planning model \cite{parkIntermittentClickPlanning2020a} improves on this weakness by providing insight on when the click event occurs. A shared problem among all these heuristics is that they fail to account for the variability in human behavior. If an attempt is made to predict the landing position of a cursor based purely on kinematic and spatial data, then the fact that a user could abort the process mid-way is not accounted for.

\subsection{Data-Driven and Learning Paradigms}

Data-driven and learning paradigms expand on purely spatial and kinetic heuristics by introducing techniques such as neural networks \cite{lePredicTouchSystemReduce2017, biswasIntentRecognitionUsing2013a} and reinforcement learning, including its use in biomechanical simulation \cite{liSelectSuggestReinforcement2022, 
oulasvirtaComputationalRationalityTheory2022, cheemaPredictingMidAirInteraction2020, chenAdaptiveModelGazebased2021, doSimulationModelIntermittently2021, moonRealtime3DTarget2024a, ikkalaBreathingLifeBiomechanical2022, miazgaLog2MotionBiomechanicalMotion2026}. Neural network based approaches work by learning from kinematic and spatial features, including the velocity of the agent, its distance from the target, the angle between the velocity vector of the agent and the position vector of the target, etc. This helps overcome the human usage variability issues associated with purely kinematic and spatial heuristics by accounting for usage patterns. However, they fail to generalize to different tasks without requiring retraining. In addition, training neural networks requires a lot of human data and is a computationally expensive task.
In contrast, reinforcement learning frames interaction as a sequential decision-making problem aimed at maximizing cumulative reward over time. These models are typically formalized using a Markov Decision Process (MDP), or a Partially Observable MDP (POMDP) when the agent only has partial sensory access to the environment. Both frameworks are defined by a tuple consisting of a state space representing the environment, a set of permissible actions, a transition function dictating the probability of moving between states, a reward signal, and a discount factor. Unlike neural network based approaches, reinforcement learning does not require pre-collected human data, since the agent instead learns by interacting with an environment.
Biomechanical simulation builds on this by pairing reinforcement learning with a physiologically grounded musculoskeletal model and physics engine, which together serve as the environment from which the agent learns. This addresses the data requirement and generalization issues of neural network based approaches, since the simulated human can generate unlimited training data across different tasks. Studies have shown that such simulations can closely replicate human performance metrics \cite{moonRealtime3DTarget2024a, miazgaLog2MotionBiomechanicalMotion2026}, with the simulated human naturally following Fitts' law \cite{fittsInformationCapacityHuman1954} and the speed-accuracy tradeoff \cite{zhaiSpeedAccuracyTradeoff2004} as an emergent property of the simulation without being hard-coded. However, the inter-user variability in these simulations is typically smaller than observed human variability. The simulations simplify real physiology, are sensitive to the choice of hyper-parameters, and carry a high computational cost to train and run.

\subsection{Control-Theoretic and Analytical Models}
Control-Theoretic and analytical models ground prediction in fundamental physics and motor-control theory rather than fitting parameters to human data. The simplest of these, the Minimum Jerk model \cite{flashCoordinationArmMovements1985}, assumes that human reaching movements minimize jerk (the rate of change of acceleration), producing smooth, bell-shaped velocity profiles without modeling feedback or noise. More sophisticated Optimal Control models frame movement as a closed-loop process in which a simulated user continuously corrects its actions based on noisy sensory feedback. Linear-Quadratic-Gaussian (LQG) control \cite{fischerOptimalFeedbackControl2022} models the body and interface as a single dynamical system with a quadratic cost function. It combines a Kalman filter to estimate velocity from noisy observations. However, LQG's reliance on linear dynamics and quadratic costs makes it unsuitable for realistic, nonlinear biomechanical models. Model Predictive Control (MPC) \cite{klarSimulatingInteractionMovements2023} addresses this by repeatedly solving a short-horizon optimization problem at each timestep rather than the entire movement at once, allowing nonlinear, physiologically grounded arm models to be used. Intermittent Control \cite{martinIntermittentControlModel2021} instead assumes that users only correct their movement intermittently, when the discrepancy between perceived and intended cursor position grows large enough, more closely matching the discrete submovement structure observed in human trajectories. Most recently, Active Inference (AIF) models \cite{klarActiveInferenceModel2025} have reframed the problem entirely: rather than hand-tuning a cost function, the agent acts to fulfill preferences while minimizing uncertainty. This approach unifies perception, belief updating, and action into a single probabilistic loop. Notably, AIF models reproduce qualitatively different pointing behavior for targets of differing difficulty using a single fixed parameter set, whereas the other Optimal Control approaches typically require separate retuning per target or user. While these models successfully ground prediction in fundamental physics and motor-control theory, they demand computationally intense optimization (e.g., solving thousands of sampled action rollouts per timestep in the case of AIF, or repeated short-horizon optimization in MPC) that is impractical for real-time deployment on low-end devices.

\section{Methodology}
To model digital interaction, it is necessary to define the boundaries of the system being observed. We apply the analogy that each target in a digital system is a gas particle. Since there may be many such particles on the screen, the boundaries must be drawn such that each cursor-target pair represents a unique system. See \cref{fig:microstate_visualization}. To derive the implications, it is necessary to list certain assumptions and definitions that will serve as a reference for the arguments in the subsequent sections.

\subsection{Assumptions and Definitions}
Gas particles have position and velocity. As a result of these properties, they have kinetic and potential energies. Since a moving cursor and target have velocity, we assign them a kinetic energy. We also assign a potential energy due to the distance between the position of the cursor and the target. This potential energy has the relationship
\begin{equation} \label{eq:force_potential}
    \mathbf{F}(\mathbf{x})=- \nabla V(\mathbf{x})
\end{equation}
An attractive force which pulls the agent towards the target is produced by the slope of the potential field. The total energy of the agent-target pair will be
$$\mathcal{H(\mathbf{v},\mathbf{x})}=T(\mathbf{v})+V(\mathbf{x}),$$
where $\mathcal{H}$ represents the Hamiltonian of the system, $T$ represents the total kinetic energy and $V$ represents the potential energy. Expressed in terms of the kinetic energies of a moving target $i$ and moving agent $a$, this formulation takes the form
$$\mathcal{H}(\mathbf{v}_a,\mathbf{v}_i,\mathbf{x})=T(\mathbf{v}_a)+T(\mathbf{v}_i)+V(\mathbf{x}).$$

\begin{definition} \label{definition:microstates}
\textbf{\textit{Microstates.}}
    The microstate of a system represents properties such as the position and velocity of every particle at a frozen moment in time. A Hamiltonian assigns energy to a microstate as a result of the velocity and position of the particles. By extension, a microstate in an information system may be represented by the position and velocities of the cursor and target.
\end{definition}

A successful interaction may occur when the system observes a microstate where an agent is at rest and its geometric boundary intersects the boundary of the target. Conversely, the cursor could still possess a velocity when the user clicks a button, but even then, intersecting the geometric boundary is required.

\begin{definition} \label{definition:mechanical_equilibrium}
\textbf{\textit{Mechanical Equilibrium.}}
    We measure the displacement between the agent and the target from their centers. At a point $\mathbf{x}=\mathbf{0}$, the centers of the agent and the target overlap. A system has reached mechanical equilibrium when the net force on an object is zero. If we hypothesize that there exists an attractive force between an agent and a target, then the center of the target may be the only point of reference where the net force is zero. Because if a target exerts an attractive force, then each pixel that composes the target may contribute to this attractive force. It is only at the center where these vectors would cancel out. However, it is important to acknowledge that this only happens for regularly shaped targets. But even then, as we will find out in the subsequent sections, pixels of different colors may exert different magnitudes of force. Having taken this ambiguity into account, we approximate the center of a target, $\mathbf{x}=\mathbf{0}$, to be the point of mechanical equilibrium.
\end{definition}

For the force $F$ to be zero, the change in the potential field, $\nabla V(\mathbf{x})$ must be zero. We assume that there is a potential minimum at $V(\mathbf{0})=0$ for the potential field around a target. Explicitly stated, the potential field has the properties \footnote{Note how we denote the output of the potential function without bold "0". The output of this function is an element of $\mathbb{R}$.},
\begin{equation*}
    V(\mathbf{0})=0 \quad\text{and}\quad\nabla V(\mathbf{0})=\mathbf{0},\quad \text{where} \quad
    V:\mathbb{R}^n\rightarrow\mathbb{R}.
\end{equation*}
In other words, $\mathbf{x}=\mathbf{0}$ represents the critical point of the potential function. Additionally, we assume that this function is continuous and differentiable everywhere. This assumption is motivated by the physical observation that the probability of interaction between a cursor and button varies continuously with distance without sudden jumps \footnote{\cref{fig:teaser_potential_well} visualizes the potential energy function.}. 

It follows that the distance vector to the boundaries of a regularly shaped target of width $W$, measured from the center $\mathbf{x}=\mathbf{0}$, can be represented as $\|\mathbf{x}\|=\frac{\mathbf{W}}{2}$. On the boundary, $V(\frac{W}{2})\ne0$ and $\nabla V(\frac{W}{2})\ne0$.

\begin{definition} \label{definition:thermal_equilibrium}
\textbf{\textit{Thermal Equilibrium.}}
    An observed system is in thermal equilibrium with another system, the heat bath, when they have the same average temperature and there is no net heat flow between the two systems. Equilibrium occurs when the systematic friction (damping) and the random thermal noise balance each other out.
\end{definition}

The core hypothesis is that an agent-target system subject to initial conditions may or may not be in thermal equilibrium initially, but might eventually relax towards thermal equilibrium over time. In the subsequent sections, we diagnose the relaxation towards equilibrium. By this analogy, the heat bath represents the random variability from user motor behavior, cognitive variability, device sensing, and measurement. The cursor and button system can be thought of as reaching thermal equilibrium when the random noise introduced by the user balances the energy lost through damping, or friction, as the cursor settles around the target.

\begin{table*}[t]
    \centering
    \begin{tabularx}{\textwidth}{llX}
        \toprule
        \textbf{Physics Concept} & \textbf{HCI Equivalent} & \textbf{Practical Meaning in this Model} \\
        \midrule
        Microstate & Frozen Frame & The exact position and velocity of the cursor and target at one specific millisecond. \\
        \addlinespace
        Thermal Energy ($k_BT$) & Observed Fluctuations & Controls how strongly differences in measurement influence the Gibbs weights. Larger values produce broader and less concentrated probabilities. \\
        \addlinespace
       Heat Bath & Interaction Environment & Represents the sources of variability, including motor noise, cognitive friction, and device effects. \\
        \addlinespace
        Harmonic Potential & Target UI Attraction & The spring-like pull a user feels toward the exact center of a target. \\
        \addlinespace
        Hamiltonian ($\mathcal{H}$) & Total System Energy & The combined kinetic energy due to movement of the cursor and the potential energy due to the pull of the target. \\
        \bottomrule
    \end{tabularx}
    \caption{A conceptual mapping of thermodynamic principles to human-computer interaction.}
    \label{tab:thermodynamic_mapping}
\end{table*}

\subsection{The Gibbs Distribution}
\begin{definition}
\textbf{\textit{Gibbs distribution.}} For a system that is in thermal equilibrium, the Gibbs distribution is a function that computes the probability of a system being in a specific state based on that state's energy and the temperature.
\begin{equation} \label{eq:gibbs_distribution}
\mathbb{P} = \frac{1}{Z} \cdot \mathrm{e}^{-\frac{\mathcal{H}}{k_BT}}
\end{equation}
Here, $\mathcal{H}$ represents the Hamiltonian, the total energy of a microstate, $k_BT$ represents the thermal energy due to temperature $T$, and the partition function $Z$ normalizes the probability.
\end{definition}

It follows from \cref{definition:thermal_equilibrium} that as the cursor begins settling around the target, the Gibbs distribution provides a probabilistic description of the possible cursor and target states based on their energies, assigning higher probabilities to states with lower total energy. For example, if the cursor is moving at a low velocity and the distance between the cursor and the target is decreasing, this means that both kinetic and potential energies are low and that the interaction is more likely. Subsequently, the Gibbs distribution assigns a higher probability based on this low kinetic and potential energy. Jaynes showed that the Gibbs distribution can also be obtained from information theory \cite{jaynesInformationTheoryStatistical1957}. In statistical mechanics, the term $k_BT$ determines how strongly the Gibbs distribution penalizes high energy microstates. In an information system, temperature is not literal; thus, $k_BT$ acts as an uncertainty scale that determines how strongly differences in energy influence the probability output. Refer to \cref{tab:thermodynamic_mapping} for a conceptual mapping of the physics terminology to HCI. If the user moves their cursor along a very noisy path or there is substantial measurement noise in the touchpad of a laptop, the term $k_BT$ will be higher. 

If a cursor is present on a screen with 3 other buttons, there will be 3 observable cursor-button systems, each with its own microstate at an instant in time. Refer to \cref{fig:microstate_visualization}. The probability weight associated with a system being in a particular microstate is
\begin{equation} \label{eq:microstate_declaration}
    \psi_i = \mathrm{e}^{-\frac{\mathcal{H}_i}{k_BT}}, \quad \psi_i \in [0, 1].
\end{equation}
This formulation describes the unnormalized probability weight of agent-target interaction given the current energy and induced noise constraints. By the boundary around each agent and target pair, the partition function $Z$ for $N$ targets is the sum of $N$ candidate weights, represented as
\begin{equation} \label{eq:partition_function}
Z = \psi_{null} + \sum_{j=1}^{N} \psi_j.
\end{equation}
The null weight represents a system without targets, where the boundary is drawn only around the agent. This null probability accounts for states in which the user does not have the intention to interact with any target.

The weighted probability of an agent interacting with a specific target, compared to $N$ other targets, is therefore 
\begin{equation} \label{eq:raw_interaction_probability}
\mathbb{P}_i = \frac{\psi_i}{\psi_{null} + \sum_{j=1}^{N} \psi_j}, \quad \mathbb{P}_i \in [0, 1].
\end{equation}

\begin{figure}
\centering
\includegraphics[width=\columnwidth]{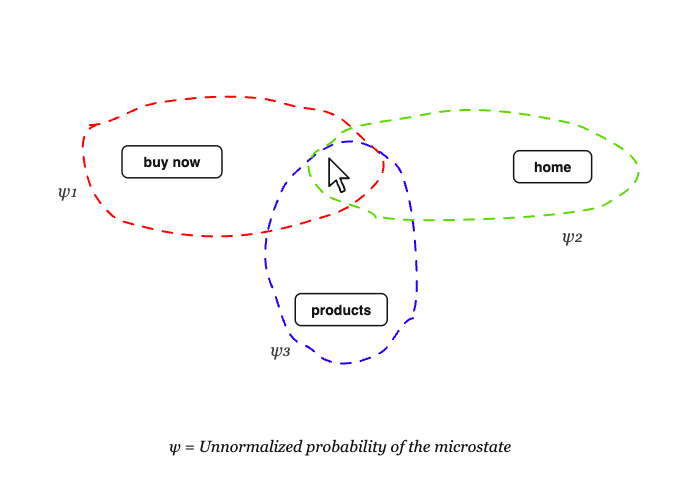}
\caption{The dotted shapes define the system boundaries for each agent-target pair. The variable $\psi_i$ represents the unnormalized probability of the respective system being in a particular microstate.}
\label{fig:microstate_visualization}
\end{figure}
 
\subsection{Interaction in Thermal Equilibrium}
\begin{proposition} \label{proposition:proposition_products}
\textbf{\textit{For an agent-target system that has relaxed towards equilibrium, the unnormalized probability weight describing the likelihood of interaction is given by }}
\begin{equation*}
    \psi_i \approx \exp\left(-\frac{\|\mathbf{v}_a\|^2}{2\sigma_{va}^2}\right) \cdot \exp\left(-\frac{\|\mathbf{v}_i\|^2}{2\sigma_{vi}^2}\right) \cdot \exp\left(-\frac{\pi\mathrm{e} \|\mathbf{x}\|^2}{W^2}\right)
\end{equation*}
\end{proposition}
\begin{proof}
For a microstate representing an agent $a$, moving with velocity $\mathbf{v}_a$, and a target $i$ moving with velocity $\mathbf{v}_i$, with a distance $\mathbf{x}$ between them, the Hamiltonian is:
\begin{equation} \label{eq:hamiltonian_derivation}
\begin{aligned}
    \mathcal{H}_i(\mathbf{v}_a,\mathbf{v}_i,\mathbf{x}) &= T(\mathbf{v}_a) + T(\mathbf{v}_i) + V(\mathbf{x}) \\
    &= \frac{1}{2} m_a \|\mathbf{v}_a\|^2 + \frac{1}{2} m_i \|\mathbf{v}_i\|^2 + V(\mathbf{x}).
\end{aligned}
\end{equation}

Substituting the Hamiltonian into \cref{eq:microstate_declaration} gives,
\begin{equation} \label{eq:microstate_hamiltonian_derivation}
\begin{aligned}
    \psi_i &= \exp\left(-\frac{\frac{1}{2} m_a\|\mathbf{v}_a\|^2 + \frac{1}{2} m_i\|\mathbf{v}_i\|^2 + V(\mathbf{x})}{k_BT}\right) \\
    &= \exp\left(-\frac{m_a\|\mathbf{v}_a\|^2}{2k_BT}\right) \cdot \exp\left(-\frac{m_i\|\mathbf{v}_i\|^2}{2k_BT}\right) \cdot \exp\left(-\frac{V(\mathbf{x})}{k_BT}\right).
\end{aligned}
\end{equation}

The term $m/k_BT$ represents a ratio of thermal energy. Using dimensional analysis to inspect the units of this term results in
$$\frac{[m]}{[k_BT]}= \frac{\mathrm{kg}}{\mathrm{kg\cdot m^2 \cdot s^{-2}}}= \frac{1}{\mathrm{m^2 \cdot s^{-2}}}.$$
Assuming that the velocity $\mathbf{v}$ in an information system is measured in $\text{px} \cdot \text{s}^{-1}$ instead of $\text{m} \cdot \text{s}^{-1}$, this enables the substitution
\begin{equation} \label{eq:sigma_propto_temp}
  \frac{m}{k_BT} = \frac{1}{\sigma_v^2}, \quad [\sigma_v^2]=\text{px}^2 \cdot \text{s}^{-2}.
\end{equation}

Any function of the form $f(x) = \exp(-a x^2)$ represents the unnormalized kernel of a Normal distribution, where $a=\frac{1}{2\sigma^2}$ and $\sigma^2$ represents the variance. Following this axiom, the term $\frac{m}{k_BT}$ in \cref{eq:microstate_hamiltonian_derivation} can be replaced by $\frac{1}{\sigma_v^2}$, the variance in velocity.

A single microstate, which represents a frozen moment in time, has no variance in velocity, and the value is exact.  However, as the system evolves over time and subsequent microstates are observed, $\sigma_v^2$ captures the variance of velocity across these successive microstates. 

Substitution back into \cref{eq:microstate_hamiltonian_derivation} yields
\begin{equation} \label{eq:microstate_kinetic_finished}
\begin{aligned}
    \psi_i &= \exp\left(-\frac{\|\mathbf{v}_a\|^2}{2\sigma_{va}^2}\right) \cdot \exp\left(-\frac{\|\mathbf{v}_i\|^2}{2\sigma_{vi}^2}\right) \cdot \exp\left(-\frac{V(\mathbf{x})}{k_BT}\right) \\
    &= \mathcal{K}_a \cdot \mathcal{K}_i \cdot \mathcal{V}.
\end{aligned}
\end{equation}
where $\mathcal{K}$ represents the exponents of kinetic energy and $\mathcal{V}$ represents the exponent of potential energy. For a function $f(\mathbf{x})$, the multivariate Taylor series around a chosen point $\mathbf{a}$ is
\begin{equation} \label{eq:taylor_formula}
f(\mathbf{x})
=
\sum_{n=0}^{\infty}
\frac{1}{n!}
\left[
\left(
(\mathbf{x}-\mathbf{a})\cdot\nabla
\right)^n
f
\right](\mathbf{a}).
\end{equation}

Following the potential energy $V(\mathbf{x})$ in
\cref{eq:microstate_kinetic_finished}, its Taylor expansion around the
point $\mathbf{x}=\mathbf{0}$ of minimum potential is
\begin{equation} \label{eq:taylor_potential}
V(\mathbf{x})
=
V(\mathbf{0})
+
\nabla V(\mathbf{0})^{\mathsf T}\mathbf{x}
+
\frac{1}{2}
\mathbf{x}^{\mathsf T}
\nabla^2V(\mathbf{0})
\mathbf{x}
+
\dots
\end{equation}

From \cref{definition:mechanical_equilibrium},
$\nabla V(\mathbf{0})=V(\mathbf{0})=0$. Therefore,
\begin{equation} \label{eq:taylor_approximation}
V(\mathbf{x})
=
\frac{1}{2}
\mathbf{x}^{\mathsf T}
\nabla^2V(\mathbf{0})
\mathbf{x}
+
\dots
\end{equation}

If the potential has the same local curvature in every spatial direction,
then $\nabla^2V(\mathbf{0})=k\mathbf{I}$. We discard the higher order terms and only consider the second degree term. Consequently, the potential function is approximated as
\begin{equation} \label{eq:harmonic_approximation}
V(\mathbf{x})
\approx
\frac{1}{2}k\mathbf{x}^{\mathsf T}\mathbf{x}
=
\frac{1}{2}k\|\mathbf{x}\|^2
\end{equation}

This approximation becomes more accurate as $\mathbf{x} \rightarrow 0$. \cref{eq:harmonic_approximation} represents the formula for elastic potential energy, called the harmonic potential. When the cursor approaches closer to a button, the user's motor-cognitive system seeks to minimize this harmonic potential. The user attempts to reach mechanical equilibrium by aiming for the point $\mathbf{x}=0$, defined in \cref{definition:mechanical_equilibrium} to be the center of the target. Any spatial deviation from this center is met with a restoring force. This formulation suggests that the agent-target system behave like a spring. Users overshoot and correct their trajectory because of the restoring force, exhibiting damped harmonic motion about the target center, as empirical studies have previously determined \cite{ziebartProbabilisticPointingTarget2012}.

Substituting \cref{eq:taylor_approximation} into the potential energy exponent of \cref{eq:microstate_kinetic_finished} results in the following
\begin{equation} \label{eq:potential_exponent}
    \mathcal{V}(\mathbf{x}) \approx \exp\left(-\frac{k\cdot \|\mathbf{x}\|^2}{2k_BT}\right).
\end{equation}
The stiffness $k$ of the spring is defined by Hooke's law as $F = k \cdot d$. Following a dimensional analysis for the term $k/k_BT$,
\begin{equation}
[F] = [k] \cdot [d] \enspace \rightarrow \enspace
   [k] = \frac{[F]}{[d]} =\frac{\mathrm{kg \cdot m \cdot s^{-2}}}{\mathrm{m}} = \mathrm{kg \cdot s^{-2}},
\end{equation}
\begin{equation}
    \frac{[k]}{[k_BT]} =\frac{\mathrm{kg \cdot s^{-2}}}{\mathrm{\mathrm{kg \cdot m^2 \cdot s^{-2}}}} = \frac{1}{\mathrm{m^{2}}}.
\end{equation}

In an information system where distance is measured in pixels ($\mathrm{px}$) rather than meters ($\mathrm{m}$)
\begin{equation} \label{eq:term_reduced}
    \frac{k}{k_BT} =\frac{1}{\sigma_{spatial}^2}, \quad [\sigma_{spatial}^2]=\text{px}^2.
\end{equation}
Here, $\sigma_{spatial}^2$ represents the variance in position over subsequent microstates. When an agent attempts an interaction at $\mathbf{x}=\mathbf{0}$ after reaching the target, this variance in position along the trajectory becomes the distribution of endpoints around $\mathbf{x}=\mathbf{0}$. Substituting \cref{eq:term_reduced} into \cref{eq:potential_exponent} gives
\begin{equation} \label{eq:potential_exponent_gaussian}
    \mathcal{V}(\mathbf{x}) \approx \exp\left(-\frac{\|\mathbf{x}\|^2}{2\sigma_{spatial}^2}\right)
\end{equation}

This formulation suggests that for a system settled in thermal equilibrium around a point of mechanical equilibrium, the distribution of endpoints is normally distributed, as previously empirically determined \cite{MacKenzie01031992, biFFittsLawModeling2013, harrisSignaldependentNoiseDetermines1998}. These studies provide us the liberty to replace $\sigma_{spatial}$ with an empirically determined approximation\footnote{We will see in the subsequent sections that the width may be accounted for in other ways and that this step of relying upon empirical substitutions might be unnecessary.}. MacKenzie's formulation \cite{MacKenzie01031992} of the standard deviation in endpoints dictates that the effective width for a target is
\begin{equation} \label{eq:mackenzie_formulation}
W = \sigma_{spatial} \sqrt{2\pi\mathrm{e}}.
\end{equation}
This approximation successfully captures $96\%$ of the interaction endpoints. Substituting \cref{eq:mackenzie_formulation} into \cref{eq:potential_exponent_gaussian} results in
\begin{equation} \label{eq:finaled_potential_component}
    \mathcal{V}(\mathbf{x};W) \approx \exp\left(-\frac{\pi\mathrm{e} \|\mathbf{x}\|^2}{W^2}\right).
\end{equation}

Finally, the unnormalized probability weight for interaction between an agent $a$ and the target $i$, at thermal equilibrium, can be represented as
\begin{equation} \label{eq:finalized_microstate_probability}
    \psi_i \approx \exp\left(-\frac{\|\mathbf{v}_a\|^2}{2\sigma_{va}^2}\right) \cdot \exp\left(-\frac{\|\mathbf{v}_i\|^2}{2\sigma_{vi}^2}\right) \cdot \exp\left(-\frac{\pi\mathrm{e} \|\mathbf{x}\|^2}{W^2}\right) \\
\end{equation}
\end{proof}
Practically, this theorem allows information systems to predict the user's target with computational efficiency. The model evaluates in
constant O(1) time by computing a joint product of probability
distributions. It must be acknowledged that this theorem assumes that thermal equilibrium has already been reached. Following \cref{definition:thermal_equilibrium}, we assume that this thermal equilibrium will start settling in close proximity to the target. Therefore, this model works best when estimating user intention to interact near the position of mechanical equilibrium, the center of the target. The formulation achieved by substituting \cref{proposition:proposition_products} into \cref{eq:raw_interaction_probability} is what we shall use to determine the target of the user, in the website prefetching experiments. 

The model may possess a weakness when predicting the target at large distances. This is due to the discarded higher order Taylor expansion terms. While this approximation may struggle for large distances, we are only concerned with the intended target of the user as they get in close proximity. During this close proximity, and before they click the target, we shall prefetch the website link. 

The velocity term in this formulation may be replaced by the projection of the velocity vector onto the line connecting the centers of the agent and the target, making the term sensitive to the orientation of movement relative to each target. The energies comprising the Hamiltonian $\mathcal{H}_i$ characterize a microstate; consequently, the velocities and distance in this formulation are instantaneous. The null probability $\psi_{null}$ is represented by a boundary that includes only the agent. It can be assumed that if an agent does not intend to interact with any target, then they might be stationary and would not expend any energy unnecessarily. Since the agent might be stationary and the system contains no targets, the kinetic and potential energies are 0. Therefore, the Hamiltonian is $\mathcal{H}=0$. Following this, the probability $\psi_{null}=1$ in \cref{eq:raw_interaction_probability}.

\subsection{Langevin dynamics and Fitts' law}
The Gibbs distribution describes the probability of different microstates once a system has relaxed towards equilibrium. However, it does not describe how the system approaches thermal equilibrium or how long this process takes. Following the target center defined to be $\mathbf{x}=\mathbf{0}$, the target occupies the interval
\begin{equation}
-\frac{W}{2}
\leq x(t)
\leq
\frac{W}{2},
\end{equation}
where $x(t)$ is a function representing the distance between the agent and target as a function of time $t$. The initial displacement is $x(0)=D$. The agent adopts a stochastic path towards the target. Langevin dynamics provides a description of how a system changes over time when it experiences random and fluctuating forces. The stochastic differential equation for these fluctuating forces is \cite{kramersBrownianMotionField1940}:
\begin{equation} \label{eq:langevin_dynamics}
m x''(t)=-\gamma x'(t)-kx(t)+\eta(t),
\end{equation}
where $m>0$ is the effective inertia, $\gamma>0$ is the damping coefficient, and $\eta(t)$ represents the instantaneous random variability acting on the movement. The heat bath represents the source of this variability and dissipation, while $\eta(t)$ represents its effect at a particular instant. In the standard equilibrium white noise model,
\begin{equation} \label{eq:langevin_noise}
\left\langle\eta(t)\right\rangle=0.
\end{equation}
The condition states that the random force has no preferred direction when averaged over comparable interactions. In this equation, when the random fluctuations balance with damping, the Gibbs distribution emerges and thermal equilibrium is established. Taking the average of \cref{eq:langevin_dynamics} and using $\langle\eta(t)\rangle=0$ gives
\begin{equation} \label{eq:mean_langevin_dynamics}
m\langle x''(t) \rangle+\gamma\langle x'(t)\rangle+k\langle x(t)\rangle=0.
\end{equation}
Thus, individual movements may remain noisy, but their mean follows the equation of a damped harmonic oscillator. Empirical studies previously performed on cursor trajectories suggest that human movement about a target involves overshooting and correcting \cite{ziebartProbabilisticPointingTarget2012}, and these oscillations graphically appear as underdamped. The solution of \cref{eq:mean_langevin_dynamics} for an underdamped oscillator can be written as
\begin{equation} \label{eq:mean_langevin_solution}
\langle x(t) \rangle=D\exp(-\lambda t)\cos(\omega_dt-\phi),
\end{equation}
Here, $\lambda$ is the rate at which the oscillation decreases, $\omega_d$ is its frequency, and $\phi$ is its starting phase. See Figure 3 for a visualization of how
the damped harmonic motion appears graphically

\begin{theorem} \label{theorem:theorem_fitts}
\textbf{\textit{For a target of width $W$ at a distance $D$ from the agent, the time taken for the amplitude of the underdamped Langevin oscillator to settle within the boundaries of the target is given by}}
\begin{equation*} \label{eq:fitts_law_recovered}
MT=a+b\log_2\left(\frac{2D}{W}\right)=\tau_e+\frac{\ln(2)}{\lambda} \log_2\left(\frac{2D}{W}\right)
\end{equation*}
\end{theorem}

\begin{proof} For oscillations to be contained within the target over time, their amplitude must be in the interval $[\frac{-W}{2},\frac{W}{2}]$. Represented as an inequality,
\begin{equation}
\frac{-W}{2} \le D e^{-\lambda  t} \le \frac{W}{2}.
\end{equation}
Solving for $t$ gives the time taken for the oscillations to settle inside the target.

\begin{equation}
\begin{aligned}
    e^{-\lambda t} &= \frac{W}{2D}\\
    t &= \frac{\ln\left(\frac{2D}{W}\right)}{\lambda}\\
    t &= \left( \frac{\ln(2)}{\lambda} \right) \log_2\left(\frac{2D}{W}\right).
\end{aligned}
\end{equation}
The total interaction interval MT must account for the time taken $\tau_e$ to initiate the action. This includes the time it takes for the user to click a button and any extra time not taken into account. Therefore, the total time for the interaction interval is
\begin{equation} \label{eq:total_int_intr_time}
MT= \tau_e +  t.
\end{equation}
After substitution, this yields
\begin{equation} \label{eq:total_int_intr_time_final}
MT= \tau_e +  \left( \frac{\ln(2)}{\lambda} \right) \log_2\left(\frac{2D}{W}\right).
\end{equation}
which is the logarithmic form of Fitts' law \cite{MacKenzie01031992}.
\end{proof}

Other studies have also connected Fitts' Law to a damped oscillator \cite{lammertSpeedaccuracyTradeoffsHuman2018}. The Gibbs distribution and Fitts' law now describe two parts of the same framework. The Gibbs distribution describes the stable probability distribution approached by the system. Langevin dynamics describes how the system relaxes towards that distribution. 

It should be noted that Fitts' Law does not describe the time taken to reach equilibrium. Instead, it is the time taken for the agent to settle inside the target boundary.

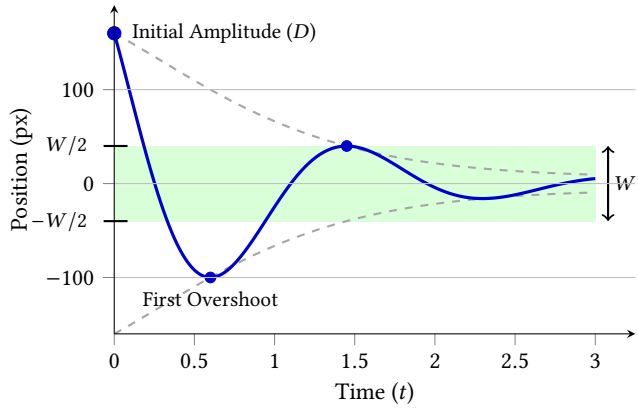
\begin{figure}[htbp]
    \centering
    \begin{tikzpicture}
    \begin{axis}[
        width=\linewidth,
        height=0.70\linewidth,
        axis x line=bottom,
        axis y line=left,
        ymin=-160,
        ymax=190,
        xmin=0,
        xmax=3.25,
        xtick={0,0.5,1,1.5,2,2.5,3},
        xlabel={Time ($t$)},
        ylabel={Position ($\textnormal{px}$)},
        ymajorgrids=true,
        grid style=solid,
        tick align=outside,
        enlargelimits=false,
        clip=true,
        clip mode=individual,
        axis on top
    ]

        \fill[green!15!white]
            (axis cs:0,-40)
            rectangle
            (axis cs:3,40);

        \draw[black, thick]
            ([xshift=-4pt]axis cs:0,40)
            --
            ([xshift=5pt]axis cs:0,40);

        \draw[black, thick]
            ([xshift=-4pt]axis cs:0,-40)
            --
            ([xshift=5pt]axis cs:0,-40);

        \node[
            anchor=east,
            font=\small
        ] at ([xshift=-7pt]axis cs:0,40)
        {$W/2$};

        \node[
            anchor=east,
            font=\small
        ] at ([xshift=-7pt]axis cs:0,-40)
        {$-W/2$};

        \addplot[
            domain=0:3,
            samples=200,
            dashed,
            gray!70,
            thick
        ] {
            160 * exp(
                -0.5835 * x
                -0.3864 * x^2
                +0.0889 * x^3
            )
        };

        \addplot[
            domain=0:3,
            samples=200,
            dashed,
            gray!70,
            thick
        ] {
            -160 * exp(
                -0.5835 * x
                -0.3864 * x^2
                +0.0889 * x^3
            )
        };

        \addplot[
            domain=0:3,
            samples=300,
            blue!80!black,
            very thick
        ] {
            199 * exp(-1.080 * x)
            * cos(deg(3.701 * x + 0.637))
        };

        \addplot[
            mark=*,
            blue!80!black,
            mark size=2.5pt
        ] coordinates {
            (0,160)
        };

        \node[
            anchor=west,
            font=\small
        ] at (axis cs:0.06,160)
        {Initial Amplitude ($D$)};

        \addplot[
            mark=*,
            blue!80!black,
            mark size=2pt
        ] coordinates {
            (0.60,-100)
        };

        \node[
            anchor=north,
            font=\small
        ] at (axis cs:0.60,-106)
        {First Overshoot};

        \addplot[
            mark=*,
            blue!80!black,
            mark size=2pt
        ] coordinates {
            (1.45,40)
        };

        \draw[<->, thick]
            (axis cs:3.08,-40)
            --
            node[
                midway,
                right,
                inner sep=2pt,
                font=\small
            ] {$W$}
            (axis cs:3.08,40);

    \end{axis}
\end{tikzpicture}
    \caption{This figure represents a target with a width of  $W$px. When the agent overshoots after each approach, the maximum overshoot distance decays and eventually gets contained within the target bounds, represented by the green segment.}
    \label{fig:damped_oscillator}
\end{figure}

\subsection{Parameterized Potential Field Curvature}
Agent and target interaction is influenced by multiple factors, such as the color contrast of a button with its environment and the clarity of the text inside it. It is a given that each of these properties has an influence on the target's attraction towards a user. We hypothesize that these parameters influence the attractive force towards the agent by modifying the steepness of the potential field. This is motivated by the relationship between force and potential energy, as seen in \cref{eq:force_potential}.

\begin{lemma}\label{lemma:lemma_parameterized_field}\textbf{\textit{The approximation of the potential field due to $n$ parameters $P$, all of which independently influence the strength of the field is
}}
\begin{equation*}
   V(\mathbf{x};P_1,...,P_n) = \frac{1}{2}(k_x)\|\mathbf{x}\|^2+\frac{1}{4} (k_{P_1}) \|\mathbf{x}\|^2 P_1^2 +\dots+ \frac{1}{4} (k_{P_n}) \|\mathbf{x}\|^2 P_n^2
\end{equation*}
\end{lemma}
\begin{proof}
    Refer to \cref{subsubsection:proof_for_lemma_parameterized_field} for the proof.
\end{proof}

For two parameters, the clarity of the label inside a button and the color contrast with its environment, this would become 
\begin{equation}
   V(\mathbf{x};L,C) = \frac{1}{2}(k_x)\|\mathbf{x}\|^2+\frac{1}{4} (k_L) \|\mathbf{x}\|^2 L^2 + \frac{1}{4} (k_C) \|\mathbf{x}\|^2 C^2
\end{equation}

As described in the proof, we set $L=0$ and $C=0$ to denote the reference values at which label clarity is extremely ambiguous and color contrast is extremely poor.

The width of a target is definitely amongst its properties. Therefore, the potential field function could easily be written as $V(x;W,L,C)$. Substituting into the Hamiltonian, this would introduce three additional multiplicative terms in \cref{proposition:proposition_products}. If one is concerned with only the width, we could also represent \cref{proposition:proposition_products} with just the exponent of width by willingly ignoring $L$ and $C$. 

However, there is a reason that the step of empirical substitution in \cref{eq:mackenzie_formulation} has been adopted in the construction of \cref{proposition:proposition_products}. It is because the width and distance share the same unit of pixels, so they may not be \textit{independent}, as required by the proof of this lemma. Another perspective is that the units may not even matter because the distance is represented by a variable and the width is a constant parameter for a target that does not change. There is additional ambiguity surrounding this lemma regarding the exact mathematical representation of label clarity and color contrast. Furthermore, no criterion has yet been formalized that distinguishes what counts as a parameter allowed to change the potential field's strength. The proof also relies upon the potential field curvature being the same in every direction. If an attempt is made to apply this lemma to regular buttons on a website, the assumption of a perfectly curved and smooth potential field is easily countered with a simple exercise:

We require that a button on a website is composed of multiple pixels. Additionally, we also require that a single pixel has a potential field as seen in \cref{fig:teaser_potential_well}. If we assemble all these pixels together to form the desired button, then their corresponding potential fields must add up as well. It is highly unlikely that any particular button will have a perfectly smooth potential field as a result of this addition process, much less all the buttons on the website. 

If we accept that this lemma only applies to single pixels rather than sets of pixels (buttons), then label clarity as a target pixel  property may not make sense at all. In that scenario, only the width and color of the pixel may make sense. Conceding to the ambiguity surrounding this lemma, we leave \cref{proposition:proposition_products} unchanged by keeping Mackenzie's formulation. Nonetheless, this lemma hints at a potentially powerful theorem that may emerge in future studies. This initial formulation serves as an important step in that direction.

In light of these insights, it reveals further details about \cref{proposition:proposition_products}. Implicitly, \cref{proposition:proposition_products} approximates the entire button with its interactive (clickable) area as a point target, the size of a pixel. The model requires the user to aim for this exact center pixel, defined to be the position of mechanical equilibrium. The output of the probability weight $\psi_i$ will be 1 if the user clicks on this exact pixel, but will decay exponentially if the user clicks further than this point. The model is able to penalize the probability output based on the inclusion of width $W$. However, this formulation does not natively understand target geometry. 

As we will see in the user study, this model will struggle when dealing with large targets. However, conventional user interfaces are built with buttons that are sufficiently small \cite{rossEpidemiologyinspiredLargescaleAnalysis2020}, for this model to serve as a robust approximation in production environments.  

\subsection{Endpoint Distribution due to Parameters}
The parametric properties of the potential field such as the color of a button or the clarity of its label might also have an impact on the distribution of endpoints. If a button has vague text or poor color contrast, an agent must have a larger endpoint distribution. The Equipartition theorem from statistical mechanics will be useful in proving this corollary and for the subsequent sections.
\begin{definition}\label{definition:equipartition}
\textbf{\textit{Equipartition theorem.}}
At thermal equilibrium, each independent quadratic term in the Hamiltonian has an average energy of $\frac{1}{2}k_BT$. For a harmonic mode in $d$ dimensions \footnote{Note that we have used $K$ to denote kinetic energy in this definition rather than $T$. This is because $T$ has been reserved for temperature within the context of this definition. Proofs calling upon this definition may replace kinetic energy with $T$ or $K$ as desired, but the notation will be clarified whenever this definition is called.},
\begin{equation*}
\mathcal{H}
=
K+V
=
\frac{1}{2}m\|\mathbf{v}\|^2
+
\frac{1}{2}k\|\mathbf{x}\|^2.
\end{equation*}
The kinetic and potential energies each contain $d$ independent quadratic terms. Therefore,
\begin{equation*}
\left\langle K\right\rangle
=
\frac{d}{2}k_BT,
\qquad
\left\langle V\right\rangle
=
\frac{d}{2}k_BT.
\end{equation*}
It follows that for an agent-target system in thermal equilibrium,
\begin{equation*}
\left\langle K\right\rangle
=
\left\langle V\right\rangle.
\end{equation*}
\end{definition}

\begin{corollary} \label{corollary:paramaterised_potential}
\textbf{\textit{At thermal equilibrium, as the magnitude of an independent target parameter $P$ increases, the endpoint distribution scales inversely with $P$.}}
\begin{equation}
\sigma_{spatial} \propto \frac{1}{P}
\end{equation}
\end{corollary}
\begin{proof}
    Refer to \cref{subsubsection:proof_for_corollary_paramaterised_potential} for the proof.
\end{proof}
This formulation suggests that increasing the label clarity $L$ reduces the standard deviation of the endpoints $\sigma_{spatial}$ and makes the agent more accurate. If a popup appears on a website while the user is attempting to click a button with a perfectly clear label, their accuracy may worsen. The user has to exert additional energy to recalculate a trajectory to the target, resulting in a noisy maneuver towards the button. Therefore, this theorem does not hold for states where the UI itself is changing around a target and thermal equilibrium may not hold.

\subsection{The Speed-Accuracy Tradeoff}
In statistical mechanics, temperature governs the distribution of kinetic and potential energies across successive microstates, dictating how physical properties such as position and velocity evolve over time. Within our framework, the temperature is a property of the information system, induced by the user's cognitive friction, motor noise and device sensor effects that uniformly scales the energy of all the agent-target pairs once the system is in thermal equilibrium. See \cref{tab:thermodynamic_mapping}. \footnote{The table may present confusion because it seems to assign the same description to the term $k_BT$ as what has been used to describe temperature $T$ in this passage. It must be noted that $k_B$ is the Boltzmann constant, which is just a scalar value, whereas $T$ is the actual variable influencing the system.}

\begin{theorem} \label{theorem:velocity_spatial_variance}
\textbf{\textit{At thermal equilibrium, the variance of endpoints is proportional to the mean squared speed.}}
\begin{equation*}
\sigma_{spatial}^2 \propto   \langle \|\mathbf{v}\|^2 \rangle 
\end{equation*}
\end{theorem}

\begin{proof}
    Refer to \cref{subsubsection:proof_for_velocity_spatial_variance} for the proof.
\end{proof}
This is consistent with Schmidt’s law \cite{schmidtMotoroutputVariabilityTheory1979}: the faster a user moves a cursor, the less spatially accurate they are. However, this is not always true. A counter example is when users flick their cursor across the screen with high velocity to click a target with great precision in desktop gaming. This further strengthens the notion that thermal equilibrium does not exist across the entire interaction and this theorem is only applicable to a system in thermal equilibrium. \cref{theorem:velocity_spatial_variance} does not hold when the user is actively exerting energy and accelerating the cursor. 

\subsection{Digital Interaction Field Theory}
The prior sections have provided insight into the different phases of the interaction life cycle. Each agent-target pair defines a potential field from the state of the interface and Langevin dynamics provide the stochastic differential equations that describe the movement towards a position of mechanical equilibrium as the system relaxes towards thermal equilibrium. In a state of thermal equilibrium, the Gibbs distribution among other theorems provide substantial modeling options. Since Fitts' Law is applicable to a variety of domains, as seen in
\begin{enumerate}
    \item Walking and foot placement \cite{druryVisuallycontrolledLegMovements1995}
    \item Gaze based interaction \cite{schuetzExplanationFittsLawlike2019}
    \item Finger touch and touchscreen typing \cite{biFFittsLawModeling2013}
    \item Brain-computer interfaces \cite{feltonEvaluationModifiedFitts2009}
    \item Tongue operated assistive input \cite{yousefiUsingFittssLaw2010}
    \item Midair gesture interaction \cite{dubePushTapDwell2022}
    \item Virtual reality object manipulation \cite{bricklerFittsLawEvaluation2020, aminiSystematicReviewFitts2025}
    \item Device tilt interaction \cite{mackenzieFittsTiltApplicationFitts2012}
    \item Haptic feedback interfaces \cite{kimAssistingVisuallyImpaired2014}
\end{enumerate}
And because Fitts' Law is a special case of Langevin dynamics with underdamped oscillations, then the differential equations governing stochastic movement have substantial potential to be generalizable with a wide variety of applications. For example, the differential equations for approach trajectories concerning critical and over-damping may be solved to yield expressions, as desired. These three \footnote{Underdamped, overdamped and critically damped oscillations are the three types being described here. We have already recovered Fitts' Law from the case of underdamped trajectories.} damping regimes exhaust the case of the harmonic potential considered in this study, but not the broader family of Langevin dynamics. Nonlinear potentials, changing targets, correlated noise, coupled coordinates, and memory effects may produce additional approach trajectories and movement time relations. 

From the counter-examples given in the previous sections regarding a departure from thermal equilibrium and a failure of prior formulations to describe certain types of interactions, a diagnostic can be conceived that provides insight into whether a system is in thermal equilibrium, and by extension, if the preceding theorems are applicable.
\begin{theorem}
\label{theorem:non_equilibrium_theorem}
\textbf{\textit{For a harmonic agent-target model, an interaction lies outside the equilibrium phase if}}
\begin{equation}
\left|\langle T\rangle-\langle V\rangle\right|>0.
\end{equation}
Equivalently,
\begin{equation}
\langle\mathcal{L}\rangle
=
\langle T\rangle-\langle V\rangle
\ne 0.
\end{equation}
\end{theorem}
\begin{proof}
Refer to \cref{subsubsection:non_equilibrium_theorem} for the proof.
\end{proof}
The Lagrangian ($\mathcal{L}$) may also offer further modeling options through the Euler-Lagrange equations. In a deployed real-time system, due to noise and measurement error, this theorem may take the form $\langle T\rangle-\langle V\rangle > \epsilon$, where $\epsilon$ represents a threshold.

\subsection{Bayesian Decision-Making}
This section and the subsequent one are independent of the thermodynamic framework presented in this study and focus on the engineering aspects before evaluation. \cref{eq:raw_interaction_probability} gives a probability output weighted against other links. However, this output alone does not provide enough information to make a rational decision on whether an action should be performed. Rather than assigning a heuristic parameter such as "Execute action if $\mathbb{P}>0.5$", the decision making process should be weighed against its cost and benefit, and not just the probability of an interaction occurring. Therefore, we utilize the following formulation:
\begin{equation} \label{eq:bayesian_desktop}
\mathbb{E}[U] = (\omega_{gain} \cdot \mathbb{P}_i) - \omega_{cost}, \quad \omega_{gain}, \omega_{cost} \in [0, 1]
\end{equation}

This decision making framework helps determine the utility of an action that has a probability $\mathbb{P}_i$ of occurring after weighing the expected gain $\omega_{gain}$ and cost $\omega_{cost}$. It provides a simple parameter that suggests executing an action if there is any utility, i.e. $\mathbb{E}[U]>0$. We are going to utilize this decision making framework by applying it to prefetching tasks on a website for the evaluation of \cref{proposition:proposition_products}.

\subsection{Application To Cursor \& Touchscreen}
To test \cref{proposition:proposition_products} and apply the model to a practical use case of prefetching tasks on a website, it is first necessary to filter out the noise present in the velocity of the cursor and that of touchscreen scrolling. We adopt a Kalman filter \cite{kalmanNewApproachLinear1960} for this requirement. It is important to note that this filter is a design choice in the frontend and is independent of the core model. Another filter may be used as desired. The Kalman filter provides the variance in the velocity measurement, which can be substituted in calculating the kinetic components of \cref{proposition:proposition_products}.

\subsubsection{1D Mobile Scrolling}
To clean up the noisy mobile scroll
velocity, we employ a Discrete 1D Kalman filter. The configurations of the filter used during evaluation are given in the Appendix \cref{subsubsection:mobile_kalman}. Unlike desktop environments with explicit cursors, mobile touchscreens lack a visible tracking agent. Since we are trying to model without using motion detectors to capture the user's finger location, we define the agent on a mobile as a stationary point. Because mobile interaction is dominated by single-handed use \cite{DesignFingersTouch}, and the optimal ergonomic resting position is slightly below the center of the screen \cite{kimNaturalThumbZone2019}, the fixed agent position $\boldsymbol{\mu}_a$ can be defined as:
\begin{equation} \label{eq:hardcoded_position}
\boldsymbol{\mu}_a = \begin{bmatrix} 0.5W_{px} \\ 0.55H_{px} \end{bmatrix}
\end{equation}
where $W_{px}$ and $H_{px}$ represent the width and height of the touchscreen device, respectively. This exact coordinate selection is derived from natural thumb zone analysis and can be adjusted as per requirements. Because we treat the agent as a stationary anchor while targets move toward it via the user's scroll action, the agent possesses zero kinetic energy. Consequently, $\mathcal{K}_a=1$ in \cref{proposition:proposition_products} throughout the interaction.

Since the interaction is constrained to one dimensional vertical scrolling, we calculate the displacement vector $\mathbf{x}$ using exclusively the vertical distance $\mathbf{x}_{y}$ between the agent and the target. We set the horizontal displacement $\mathbf{x}_{x}=0$.

However, this causes horizontally adjacent targets to have the same distance from the agent. Under the normalized partition function, \cref{eq:raw_interaction_probability}, these targets occupying the same horizontal position (since $\mathbf{x}_{x}=0$) would mathematically dilute the normalized probability $\mathbb{P}$ among themselves. This probability dilution would subsequently penalize the expected utility in the Bayesian decision-making framework (\cref{eq:bayesian_desktop}). 

To correct for 1D, mobile evaluations strictly utilize the unnormalized microstate probability $\psi_i$. The Bayesian decision-making framework for mobile interaction is therefore adapted as
\begin{equation} \label{eq:bayesian_mobile}
\mathbb{E}[U] = (\omega_{gain} \cdot \psi_i) - \omega_{cost}, \quad \omega_{gain}, \omega_{cost} \in [0, 1]
\end{equation}

\subsubsection{2D Cursor Movement}
To clean up the noisy cursor velocity, we employ a Discrete 2D Kalman filter. The configurations of the filter used during evaluation are given in the Appendix \cref{subsubsection:destop_kalman}.

For cursor based interaction, the agent is represented by the cursor itself. The cursor will have a kinetic energy due to its velocity. If targets are moving, such as when the screen is scrolling, then the targets will also have kinetic energy. Therefore, \cref{proposition:proposition_products} and \cref{eq:raw_interaction_probability} apply without any special considerations. 

\section{Evaluation}
In the methodology, we made a hypothesis in the proof of \cref{lemma:lemma_parameterized_field} that target parameters independently influence the strength of the potential field. On the basis of this \textit{independence}, we discarded the coupling terms in the multivariate Taylor expansion. This evaluation section must test the independence of these parameters. Additionally, we must also test whether the color contrast of a button and the clarity of its label have an effect on the accuracy of a user attempting to click it, from \cref{corollary:paramaterised_potential}. Furthermore, we must test the efficiency of the target acquisition model described in \cref{proposition:proposition_products} against competitors. 

We recruited 30 participants (22 male, 8 female). Each participant took part in all three experiments. The participants had a mean age of 23.5 years (SD = 2.8, range:
19–32). All participants were daily users of mobile touchscreens and laptops. None of the participants reported having visual or motor impairments.

\subsubsection{Ethics Statement}
The study received approval from the relevant institutional ethics committee before data collection. Participants volunteered, provided informed consent, and could withdraw at any time. No direct personal identifiers were recorded.

\subsection{Evaluation of \cref{proposition:proposition_products}}
In order to simulate a real environment with network delays and to enable participants to run the model on their own devices, \cref{proposition:proposition_products} was programmed into an e-commerce clothing website and deployed on Google Cloud Run using Docker. The pages of the website were prefetched when the expected utility was greater than zero in \cref{eq:bayesian_desktop}. The weights of \cref{eq:bayesian_desktop} would ideally be determined using performance metrics from services such as Google Analytics. Pages that are viewed more often, would get prefetched more often and therefore should have higher $\omega_{gain}$ and lower $\omega_{cost}$. This decision making process helps avoid prefetching pages that statistically have a lower chance of being visited according to usage metrics, thereby saving on compute cost. Usage metrics were not available as the website was a testbed and did not have real-world users, so all the links on the website were assigned $\omega_{gain}=0.5$ and $\omega_{cost}=0.1$. Each participant was instructed to use the website naturally and browse the products as if they intended to buy them. The website interface was designed to look professional and was built on NextJS for the single page app experience. Images of the website can be found in the appendix \cref{fig:experimental_setup_websites}. Additional optimization techniques were implemented to fully simulate a real production e-commerce website experience, such as
\begin{enumerate}
    \item CDNs for image delivery
    
    \item Lazy loading parts of the page as the user scrolled down

    \item Backend API caching using redis
    
    \item Page swap loaders. Although in most cases this became redundant as the high prefetch hit rate eliminated loading delays
\end{enumerate}
Auto-scaling and load balancing were implemented to ensure that there was no server lag as a result of multiple participants doing the experiments at the same time. To test \cref{proposition:proposition_products} in the worst case scenario, the website was designed to have links in close proximity to each other, since most competing models struggle in cluttered interfaces. The website had 110 uploaded products and browsable categories for Men, Women, Girls, Boys, and Babies. To test the prediction accuracy of \cref{proposition:proposition_products} against other heuristics, we chose two baselines, ForesightJS and prefetch all visible links. Foresight functions as a spatial and kinematic heuristic baseline which projects the cursor trajectory forward to estimate the landing position, while fetch all visible links provides a standard maximum cost baseline.

\textbf{\textit{Experiment 1: Cursor Prefetching.}} Each participant used the website on their own laptops, which included both Windows and Mac operating systems. In order to begin the experiment, the participants pressed the space key on their keyboards and then used the website naturally. The website prefetched links whenever $\mathbb{E}[U]>0$ from \cref{eq:bayesian_desktop}. Both ForesightJS and a fetch all visible links baseline were also running in the background. When the participants had decided that their browsing session was complete, they pressed the space key again to export a CSV file that contained the usage metrics. All participants used their laptop touchpads to interact with the website. 
    
\textbf{\textit{Experiment 2: Mobile Prefetching.}} Each participant used the website on their own mobile phones, which included both Android and iOS operating systems. As soon as the website opened, they tapped a button on the website to begin the experiment and had to press the same button again to export the usage metrics as a CSV file. The website prefetched links whenever $\mathbb{E}[U]>0$ from \cref{eq:bayesian_mobile}. Since ForesightJS defaults to fetching visible links on a mobile screen, this experiment only had the fetch all visible links as a baseline comparison.

The screen resolutions of the participants' mobile devices and laptops varied, but \cref{proposition:proposition_products} adapted without any special consideration or modifications. The exported file included usage metrics such as the 
\begin{enumerate}
    \item Normalized probabilities of the top five links at the moment of the click.
    \item The unnormalized probabilities of the targets preceeding the 1.5 seconds (segmented into 30 ms time bins) before the click.
    \item If the clicked link was prefetched.
    \item What prefetched links were not clicked.
\end{enumerate}
Along with the metrics for \cref{proposition:proposition_products}, these metrics were also reported for ForesightJS and the fetch all visible links baseline.

\subsection{Evaluation of \cref{lemma:lemma_parameterized_field} and \cref{corollary:paramaterised_potential}}
\textit{\textbf{Experiment 3:}} Tests the distribution of endpoints as a result of target properties and whether these parameters are statistically independent. To perform the test in a controlled environment, we created a simple one-page website and participant testing was conducted sequentially on localhost using the evaluator's laptop, a 14-inch MacBook M3 Pro. A mouse, Razer Viper Mini configured to 8000Hz, was connected to the MacBook for the participants to use. Pictures of this testbed website are in the appendix \cref{fig:experimental_setup_exp3}. To begin the experiment, participants had to click a start button on the website. The website had a rectangular area with a circular green point in the center. At any given moment, there would only be one target button in the rectangular area and the participant would have to click it. Upon clicking the target button, the user had to reset their mouse position by moving their cursor to the center of the rectangular area and clicking the green circle. Upon clicking the green circle, a new target appeared at a random angle with a fixed center-to-center distance of 280px. We employed a $2 \times 2$ within-subjects experimental design. There were 4 categories of buttons

\begin{enumerate}
    \item High Contrast Button with Clear Label
    
    \item High Contrast Button with Ambiguous Label

    \item Low Contrast Button with Clear Label

    \item Low Contrast Button with Ambiguous Label
\end{enumerate}
The experiment was carried out in 4 phases. Each phase had one of these four categories of the buttons. The user had to click 20 buttons from each phase in order to progress to the next phase. With 4 phases, 20 buttons, and 30 participants, there was a total of 2400 recorded data points. The exported files contained the standard deviation $\sigma_{spatial}$ of the endpoints, for each button category. For the clear label category, the text "Click" was present inside the buttons. For the ambiguous label, it was "\#@?.5!(-". \cref{fig:experimental_setup_exp3} in the appendix shows the picture of the testbed website and the different button categories used in this experiment.

\section{Results}
For the three experiments which were conducted, four graphs were plotted. For Experiment 1 and 2, \cref{fig:timeline_competition} and \cref{fig:efficiency} illustrate the results. Conversely, \cref{fig:sigma_bars} and \cref{fig:interaction} plot the results for Experiment 3. 

\subsection{Results of Experiment 1 and 2}
\subsubsection{Probability of Interaction}
The left subplot in \cref{fig:timeline_competition} represents the unnormalized probability that increased over the 1.5 seconds leading up to the link being clicked. The cursor interaction curve peaked at around 0.85 probability at the moment the click took place, while the mobile touch screen interaction curve peaked at almost 0.7 probability. The mobile touch screen curve started off at a higher probability of roughly 0.25 at the 1.5 second mark. The target links on mobile triggered prefetches at a median lead time of 881 ms prior to the click, compared to 307 ms for desktop cursor interactions. 

The center subplot in \cref{fig:timeline_competition} illustrates the average normalized probability in a desktop environment of the 5 links with the highest probabilities, at the instant of user click. The normalized probability of the clicked link was 24.6 times higher than the second ranked link. 

The right subplot details the unnormalized probabilities on a mobile interface of the top five predicted links, at the moment of the click. The engine successfully ranked the target link highest, though its probability was only 1.1 times greater than that of the second-ranked link.

\subsubsection{Trade-off and Competitor Analysis}
The left subplot in \cref{fig:efficiency} illustrates the desktop trade-off. Each browsing session of a single participant is represented by the 3 dots connected with dotted lines. Values further along the x-axis indicate a less efficient fetch:click ratio. The proposed model achieved a fetch:click ratio of 1.37, while ForesightJS achieved a fetch:click ratio of 2.00, and Viewport All achieved 3.45. The fetch:click ratio was calculated by defining a ratio of (total prefetches):(total clicks). ForesightJS had a 46\% $[(2.00-1.37)\div1.37]$ higher fetch:click ratio than the proposed model, while Viewport All had a 152\% $[(3.45-1.37)\div1.37]$ higher fetch:click ratio. Perfectly horizontal lines represent a browsing session in which all the clicked links were prefetched prior to the click event by all the models. Diagonal lines represent the scenario in which the clicked link was not prefetched. The proposed model occasionally failed to prefetch the target, as indicated by diagonal lines, yielding a final accuracy of 98.1\%. Meanwhile, ForesightJS and Viewport All achieved an accuracy of 100\% with all clicked links being prefetched. 

The right subplot in \cref{fig:efficiency} illustrates the mobile trade-off. The proposed model achieved a fetch:click ratio of 1.75 while maintaining an accuracy of 98.0\%. By comparison, the Viewport All achieved an accuracy of 100\% but a fetch:click ratio of 2.82. Viewport All had a 61\% $[(2.82-1.75)\div1.75]$ higher fetch:click ratio than the proposed model.

\begin{figure*}[htbp]
    \centering
    \includegraphics[width=\linewidth]{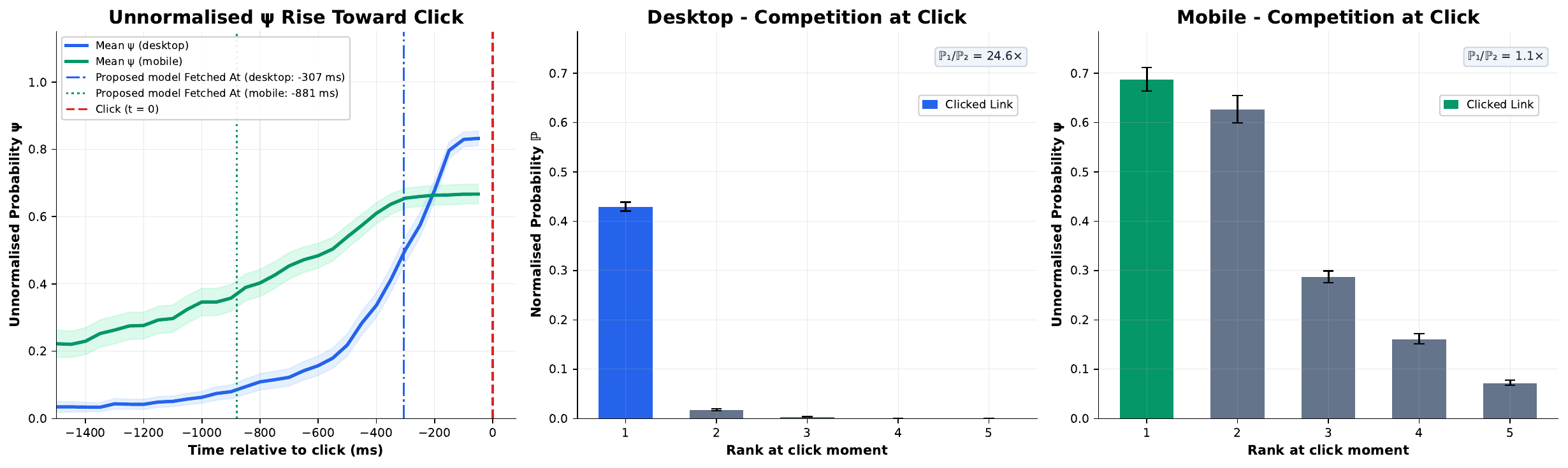}
    \caption{Left: Unnormalized probability rises prior to the click, triggering prefetches with significant lead time. Error bars represent 95\% Confidence Intervals. Center: normalized desktop probabilities. Right: unnormalized mobile probabilities. In both cases, the clicked link is ranked highest.}
    \label{fig:timeline_competition}
\end{figure*}

\begin{figure*}[t]
    \centering
    \includegraphics[width=\linewidth]{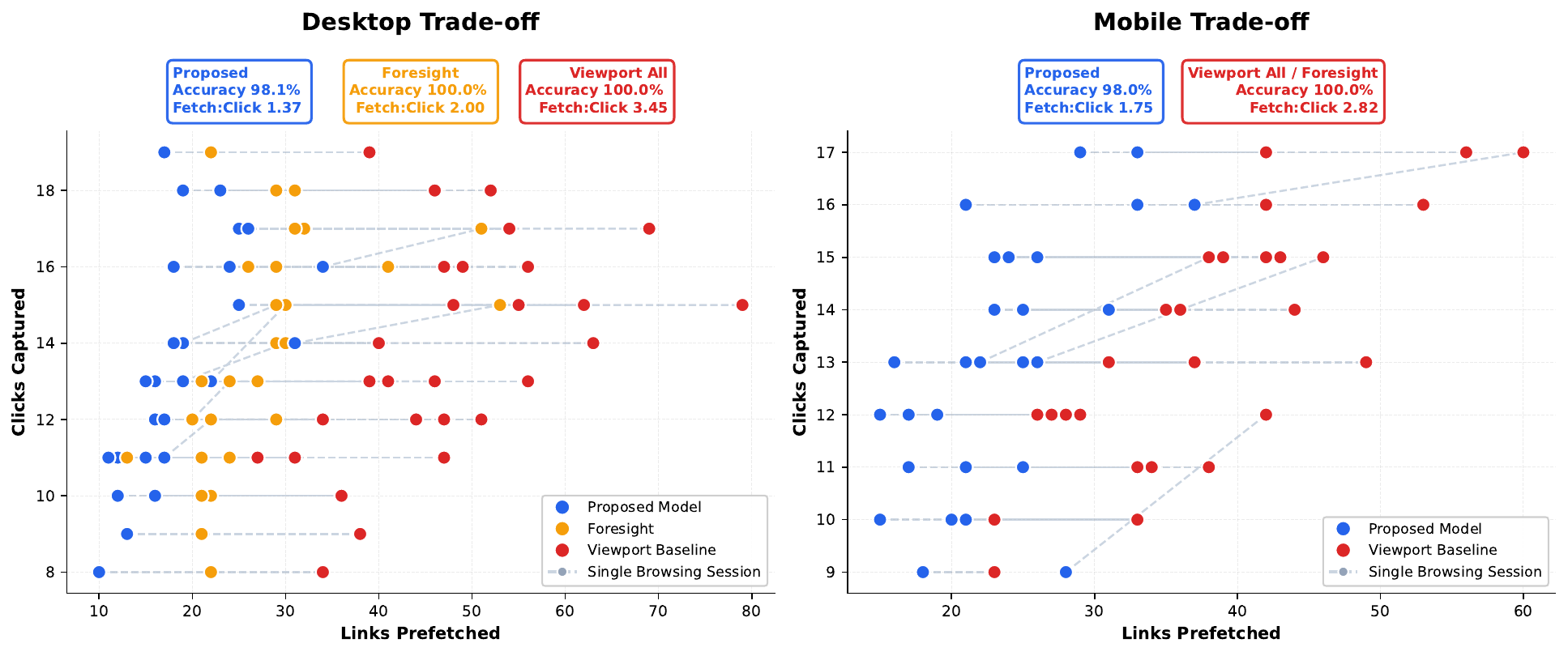}
    \caption{Prefetch strategy trade-off. The Proposed Model captures nearly all intended clicks with 98\% accuracy while drastically reducing wasted bandwidth (Fetch:Click ratio) compared to baselines.}
    \label{fig:efficiency}
\end{figure*}

\subsection{Results of Experiment 3}
To ensure strict experimental control over the 280px distance, the standard deviation of endpoints was analyzed exclusively using trials with exactly a 280px starting distance. Because repeated-measures ANOVAs require a perfectly balanced design, four participants who lacked at least one valid 280px trial across all four button categories were excluded, resulting in a final sample of $N = 26$ participants. As shown in \cref{fig:sigma_bars}, the high-contrast button with a clear label resulted in the lowest endpoint standard deviation ($\sigma_{spatial} \approx 10.5$px). Conversely, the low-contrast button with an ambiguous label produced the highest standard deviation ($\sigma_{spatial} \approx 13.0$px). A two-way repeated-measures ANOVA revealed a significant main effect of contrast on endpoint standard deviation ($F(1, 25) = 7.55, p = .011$), as well as a statistically significant main effect of label clarity ($F(1, 25) = 35.62, p < .001$). The interaction plot (\cref{fig:interaction}, right) indicates there was no significant interaction effect between contrast and label clarity ($F(1, 25) = 0.01, p = .936$), suggesting that these parameters may operate independently.

\begin{figure}[t]
    \centering
    \includegraphics[width=0.9\columnwidth]{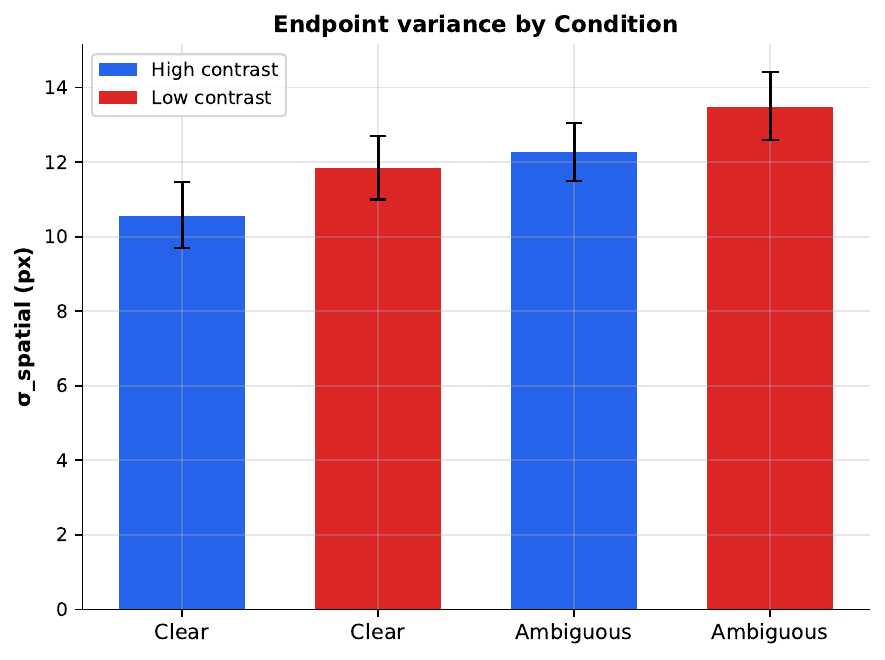}
    \caption{Overall endpoint standard deviation ($\sigma_{spatial}$) grouped by button condition. Higher contrast and clearer labels generally resulted in tighter endpoint distributions.}
    \label{fig:sigma_bars}
\end{figure}

\begin{figure*}[t]
    \centering
    \includegraphics[width=\linewidth]{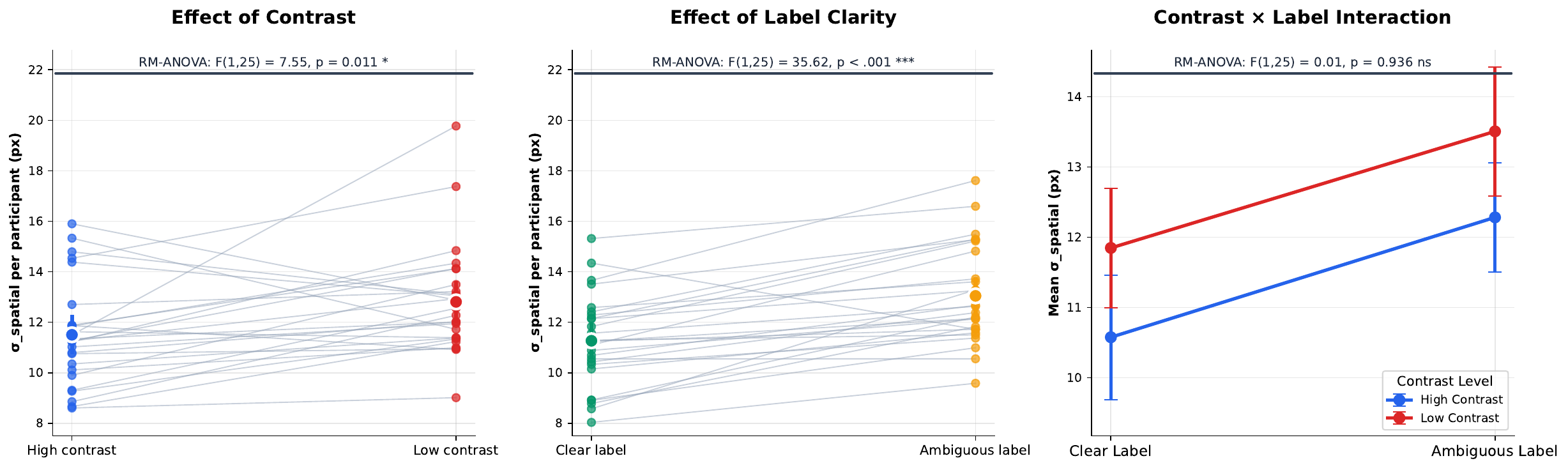}
    \caption{Within-subject effects of button parameters on endpoint standard deviation. Left \& Center: High contrast and clear labels both reduce standard deviation of endpoints. Right: The interaction plot demonstrates contrast and label clarity have no interaction with each other and may independently influence the accuracy of the user.}
    \label{fig:interaction}
\end{figure*}

\section{Discussion}
For cursor based evaluation, ForesightJS had a 46\% higher fetch:click ratio than the proposed model, while Viewport All had a 152\% higher ratio. On mobile, Viewport All had a 61\% higher fetch:click ratio. These metrics demonstrate that the proposed model is substantially more efficient than purely kinematic heuristics while maintaining the same level of computational efficiency, requiring just multiplication to compute. Despite the limitations discussed below, the results provide consistent empirical support for the theory across desktop, mobile, and controlled pointing experiments. Further studies with broader samples and interfaces can establish how widely the thermodynamic framework generalizes.

\subsection{Threats to Validity}
The study had 30 participants with a limited age range. The participants also used different devices in Experiments 1 and 2. Learning or fatigue may have affected the results. The utility weights from \cref{eq:bayesian_desktop,eq:bayesian_mobile} were fixed and other values were not tested.

Four participants were excluded from the final ANOVA shown in \cref{fig:interaction} because they did not have valid trials in all four conditions of the $2 \times 2$ within-subjects design. The evaluation used one website and one controlled pointing task.

\subsection{Review of Results}
At 1.5 seconds prior to the click, mobile interactions exhibited a higher probability than cursor-based interactions, as illustrated by the left subplot in \cref{fig:timeline_competition}. This was due to users scrolling the screen and spending some time observing the newly appeared elements on the screen. Consequently, this caused the 1.5 seconds preceeding the click event to have a higher probability. For interaction using a cursor, the probability at the 1.5 seconds preceeding the click was lower as a result of the participants' cursors being sufficiently far from the center of the clicked link. The high probability at $t=0$ ms further justifies the harmonic potential having an equilibrium position on $\mathbf{x}=\mathbf{0}$, at roughly the center of the links, towards which the users aimed. The probability of cursor based interaction is not $1$ at $t=0$ ms because of the distribution of endpoints around the center of the links, since $\mathbf{x}=\mathbf{0}$ is the point with maximum probability. In the case of interaction on mobile, users were not aligning the center of the link with the hard-coded agent position from \cref{eq:hardcoded_position}, leading to a slightly lower probability at the $t=0$ ms mark. Since users aimed for centers on cursor based interaction, the normalized probability was heavily in favor of the clicked link, which explains the 24.6 times domination over the second ranked link in the middle subplot of \cref{fig:timeline_competition}. Since mobile based interaction used an unnormalized probability, the probability difference between the ranking of the links was much lower than on desktop.

The 98.1\% accuracy of the proposed model, as illustrated by the left subplot of \cref{fig:efficiency}, was due to users sometimes not clicking near the center of extremely large targets, such as banners. Since the potential component $\mathcal{V}$ has a low probability at the edges of large targets, this resulted in the failed prediction of the target. Similarly, for mobile based interaction, the inaccuracy is a result of users sometimes clicking links towards the edges of the screen, whose centers were far away from the hard-coded agent position, and in certain cases, outside the viewport. 

The parallel lines observed in the right subplot of \cref{fig:interaction} demonstrate that the target parameters had no interaction with each other. This provides some validation for the assumption in the proof of \cref{lemma:lemma_parameterized_field}, where the coupled terms were set to zero. In general, label clarity had a greater impact on user inaccuracy than color contrast. This suggests that in the process of making contact, decoding and understanding textual information makes the user more inaccurate than differentiating colors.

\subsection{Limitations of the Target Acquisition Model}
An important finding of \cref{fig:efficiency} is that the model starts to break down when faced with large targets. The testbed e-commerce website (\cref{fig:experimental_setup_websites}) had such large targets, resulting in 98.1\% accuracy instead of 100\%. For an agent at rest at the boundaries of a large target, there will be zero kinetic energy and non-zero potential energy, causing a non-zero averaged Lagrangian as per \cref{theorem:non_equilibrium_theorem}. For the proposed harmonic agent-target model, the non-zero averaged Lagrangian would explain why the Gibbs distribution, a principle of systems in thermal equilibrium, may struggle to predict user interaction for large targets. Large targets also allow selection over large distances from the center, and this causes another issue with the Taylor approximation with discarded higher order terms in the proof of \cref{proposition:proposition_products}.

\section{Future Work}
There are no mathematical restrictions that prevent \cref{proposition:proposition_products} from being applicable beyond the evaluated website and prefetching use case. It should be applicable to 3D use cases, including gaze, virtual reality, mid-air gestures, and brain-computer interfaces. However, this has yet to be fully evaluated. Additionally, the proposed framework provides options to model the different types of approach trajectories, such as critically damped and overdamped, which have yet to be formalized and applied to practical scenarios. 

Machine learning approaches can model learned behavior more accurately by incorporating theoretical constraints such as equilibrium, non-equilibrium dynamics, and distinct agent trajectories. Furthermore, models can be trained to  account for user experience across varying button attributes, including label clarity, color contrast, and border radius. 

As the experimental results show, there was no observed interaction between label clarity and color contrast. However, future work must attempt to formalize the ambiguity in what is counted as a target parameter and how the specified property is quantified. Further extensions to the theory could investigate nonlinear and time varying potentials.

\section{Conclusion}
This study introduced a theory grounded in thermodynamics for human-computer interaction. The framework models agent-target interaction by assigning kinetic energy from velocities and potential energy from distance. The framework views the entire interaction lifecycle to be in regimes of thermodynamic equilibrium and non-equilibrium, where Langevin dynamics provide the differential equations that describe the movement of an agent as it approaches the target, and the Gibbs distribution describes the state of thermal equilibrium which the system approaches. The theory successfully recovered Fitts' Law by suggesting that it is a special case of Langevin dynamics for an approach trajectory resembling that of an underdamped oscillator. Furthermore, the theory provides insight that suggests each property of the target, such as its color contrast, may independently influence the strength of the attractive force and impact the accuracy of the user. 

The target acquisition model, presented as a derived proposition from the framework, was evaluated across 30 participants. It achieved 98.1\% desktop target prediction accuracy with a fetch ratio of 1.37. For mobile, it achieved 98.0\% accuracy with a ratio of 1.75. Another study was conducted to evaluate the impact of button color contrast and the clarity of its label on user accuracy, which found significant effects of label clarity and color contrast on endpoint variability.

The target acquisition model had a weakness when dealing with large targets in the user study, as anticipated by the theory, due to the  discarded higher order Taylor expansion terms and the thermal equilibrium assumption. Future work should extend the framework across broader interfaces, movement regimes and further develop the mathematical formalism.

\bibliographystyle{ACM-Reference-Format}
\bibliography{references}

\section{Appendix}
\subsubsection{1D Discrete Kalman Filter} \label{subsubsection:mobile_kalman}
Filter configuration for mobile vertical scrolling that was used in the production website during the evaluation.
\begin{enumerate}
    \item \textit{State Vector \& Observation Matrix.} $x_t = \begin{bmatrix} y_t & v_t \end{bmatrix}^T$. Here, $y_t$ represents the vertical position and $v_t$ represents the velocity of the scroll. The Observation Matrix is set to $\mathbf{H} = \begin{bmatrix} 1 & 0 \end{bmatrix}$.
 
    \item \textit{Measurement Noise $\mathbf{R}$ \& Process Noise $\mathbf{Q}$.} 
\begin{equation*}
        \mathbf{R} = 100
\qquad
\mathbf{Q} = \begin{bmatrix} 1 & 0 \\ 0 & 0.1 \end{bmatrix}
    \end{equation*}
 
    \item \textit{State Error Covariance Matrix.} 
    $$\mathbf{P}_t = \begin{bmatrix}\sigma_y^2 & \sigma_{yv} \\ \sigma_{vy} & \sigma_v^2 \end{bmatrix},\quad\mathbf{P}_0 = \begin{bmatrix}10 & 0 \\ 0 & 10 \end{bmatrix}$$
    The velocity variance $\sigma_v^2$ is applicable within the kinematic exponents of \cref{eq:finalized_microstate_probability}
\end{enumerate}

\subsubsection{2D Discrete Kalman Filter} \label{subsubsection:destop_kalman}
Filter configuration for the cursor that was used in the production website during the evaluation.

\begin{enumerate}
    \item \textit{State Vector \& Observation Matrix.} $x_t =\begin{bmatrix} x_t & y_t & v_{xt} & v_{yt} \end{bmatrix}^T$. Here, ($x_t,y_t$) corresponds to the cursor position coordinates and $(v_{xt},v_{yt})$ their respective velocity components. The Observation Matrix is $$\mathbf{H} = 
    \begin{bmatrix} 1 & 0 & 0 & 0 \\ 0 & 1 & 0 & 0 \end{bmatrix}$$.
 
    \item \textit{Measurement Noise $\mathbf{R}$ \& Process Noise $\mathbf{Q}$.} 
    \begin{equation*}
        \mathbf{R} =
\begin{bmatrix}
300 & 0 \\
0 & 300
\end{bmatrix}
\qquad
\mathbf{Q} =
\begin{bmatrix}
0.1 & 0 & 0 & 0 \\
0 & 0.1 & 0 & 0 \\
0 & 0 & 0.1 & 0 \\
0 & 0 & 0 & 0.1
\end{bmatrix}
    \end{equation*}
 
    \item \textit{State Error Covariance Matrix.} \begin{equation*}
        \mathbf{P}_t = \begin{bmatrix}
        \sigma_x^2 & \sigma_{xy} & \sigma_{x v_x} & \sigma_{x v_y} \\
        \sigma_{yx} & \sigma_y^2 & \sigma_{y v_x} & \sigma_{y v_y} \\
        \sigma_{v_x x} & \sigma_{v_x y} & \sigma_{v_x}^2 & \sigma_{v_x v_y} \\
        \sigma_{v_y x} & \sigma_{v_y y} & \sigma_{v_y v_x} & \sigma_{v_y}^2
    \end{bmatrix} \quad \mathbf{P}_0 = \begin{bmatrix}
        100 & 0 & 0 & 0 \\
        0 & 100 & 0 & 0 \\
        0 & 0 & 100 & 0 \\
        0 & 0 & 0 & 100
    \end{bmatrix}
    \end{equation*} The velocity variance can be computed as 
    \begin{equation*}
        \sigma_v^2=\sigma_{v_x}^2 + \sigma_{v_y}^2
    \end{equation*}
\end{enumerate}

\subsection*{Proofs}
\subsubsection{Proof for \cref{lemma:lemma_parameterized_field}} \label{subsubsection:proof_for_lemma_parameterized_field}
\begin{proof}
Let $\mathbf{x}$ represent the displacement vector between the agent and the target. Let $L$ represent label clarity and let $C$ represent color contrast. The reference values are $\mathbf{x}=\mathbf{0}$, $L=0$ describing a button with zero clarity, and $C=0$ describing a button with zero contrast. Multivariate Taylor approximation of the potential field with respect to $\mathbf{x}$ about $\mathbf{x}=\mathbf{0}$ gives
\begin{equation}
\begin{aligned}
V(\mathbf{x};L,C)
&\approx
V(\mathbf{0};L,C)
+
\nabla V(\mathbf{0};L,C)^{\mathsf T}\mathbf{x}\\
&\quad+
\frac{1}{2}
\mathbf{x}^{\mathsf T}
\nabla^2V(\mathbf{0};L,C)
\mathbf{x}.
\end{aligned}
\end{equation}

From \cref{definition:mechanical_equilibrium}, at the target center,
\begin{equation}
V(\mathbf{0};L,C)=0\quad \text{and}
\quad
\nabla V(\mathbf{0};L,C)=\mathbf{0}.
\end{equation}

We assume that the potential field has the same curvature in every spatial direction. Therefore,
\begin{equation}
\nabla^2V(\mathbf{0};L,C)
=
k(\mathbf{0};L,C)\mathbf{I},
\end{equation}
where $\mathbf{I}$ is the identity matrix. The potential field is consequently
\begin{equation}
V(\mathbf{x};L,C)
\approx
\frac{1}{2}
k(\mathbf{0};L,C)
\|\mathbf{x}\|^2.
\end{equation}

Taylor approximation of $k(\mathbf{0};L,C)$ about $L=0$ and $C=0$ gives
\begin{equation}
\begin{aligned}
k(\mathbf{0};L,C)
&\approx
k(\mathbf{0};0,0)\\
&\quad+
\left(
\frac{\partial k(\mathbf{0};0,0)}{\partial L}L
+
\frac{\partial k(\mathbf{0};0,0)}{\partial C}C
\right)\\
&\quad+
\frac{1}{2}
\left(
\frac{\partial^2k(\mathbf{0};0,0)}{\partial L^2}L^2
+
2\frac{\partial^2k(\mathbf{0};0,0)}
{\partial L\partial C}LC
+
\frac{\partial^2k(\mathbf{0};0,0)}{\partial C^2}C^2
\right).
\end{aligned}
\end{equation}

We assume that the linear terms are zero
\begin{equation}
\frac{\partial k(\mathbf{0};0,0)}{\partial L}
=
\frac{\partial k(\mathbf{0};0,0)}{\partial C}
=
0.
\end{equation}
and also hypothesize that the effects of $L$ and $C$ are additive, so the coupling term is zero:
\begin{equation}
\frac{\partial^2k(\mathbf{0};0,0)}
{\partial L\partial C}
=
0.
\end{equation}
This hypothesis of nonexistent coupling is tested in the evaluation. Define
\begin{equation}
\begin{aligned}
k_x
&=
k(\mathbf{0};0,0),\\
k_L
&=
\frac{\partial^2k(\mathbf{0};0,0)}{\partial L^2},\\
k_C
&=
\frac{\partial^2k(\mathbf{0};0,0)}{\partial C^2}.
\end{aligned}
\end{equation}

The curvature is therefore approximated by
\begin{equation}
k(\mathbf{0};L,C)
\approx
k_x
+
\frac{1}{2}k_LL^2
+
\frac{1}{2}k_CC^2.
\end{equation}

Substituting this expression into the potential energy function gives
\begin{equation}
V(\mathbf{x};L,C)
\approx
\frac{1}{2}k_x\|\mathbf{x}\|^2
+
\frac{1}{4}k_L\|\mathbf{x}\|^2L^2
+
\frac{1}{4}k_C\|\mathbf{x}\|^2C^2.
\end{equation}

Repeating the same expansion for $n$ scalar parameters, while assuming that all linear and mixed terms are zero, gives
\begin{equation}
V(\mathbf{x};P_1,\ldots,P_n)
\approx
\frac{1}{2}k_x\|\mathbf{x}\|^2
+
\frac{1}{4}\|\mathbf{x}\|^2
\sum_{j=1}^{n}k_{P_j}P_j^2.
\end{equation}
\end{proof}

\subsubsection{Proof for \cref{theorem:velocity_spatial_variance}}
\label{subsubsection:proof_for_velocity_spatial_variance}

\begin{proof}
For the Hamiltonian \footnote{Note that we have used $K$ to denote kinetic energy in this proof rather than $T$. This proof reserves $T$ for temperature.}
\begin{equation}
\mathcal{H}=K+V,
\end{equation}
\cref{definition:equipartition} gives,
\begin{equation}
\langle K\rangle\propto T
\quad\text{and}\quad
\langle V\rangle\propto T.
\end{equation}

From \cref{eq:term_reduced}, the spatial variance also scales with temperature:
\begin{equation}
\sigma_{\mathrm{spatial}}^2\propto T.
\end{equation}

The average kinetic energy is
\begin{equation}
\begin{aligned}
\langle K\rangle
&=
\frac{1}{2}m
\left\langle\|\mathbf{v}\|^2\right\rangle, \\
\left\langle\|\mathbf{v}\|^2\right\rangle
&\propto
\langle K\rangle
\propto T.
\end{aligned}
\end{equation}

Since both spatial variance and the second velocity moment scale with temperature,
\begin{equation}
\sigma_{\mathrm{spatial}}^2
\propto
\left\langle\|\mathbf{v}\|^2\right\rangle.
\end{equation}
\end{proof}

\subsubsection{Proof for \cref{corollary:paramaterised_potential}} \label{subsubsection:proof_for_corollary_paramaterised_potential}
\begin{proof}
It follows from \cref{definition:equipartition} and
\cref{lemma:lemma_parameterized_field} that
\begin{equation}
\begin{aligned}
\left\langle V(\mathbf{x};L)\right\rangle
&=
\left\langle
\frac{1}{2}k(\mathbf{0};L)\|\mathbf{x}\|^2
\right\rangle
=
\frac{d}{2}k_BT,\\
\left\langle\|\mathbf{x}\|^2\right\rangle
&=
\frac{dk_BT}{k(\mathbf{0};L)}.
\end{aligned}
\end{equation}

Since we have approximated $V(\mathbf{x};L)$ around $\mathbf{x}=\mathbf{0}$, applying the identity
\begin{equation}
\left\langle\|\mathbf{x}\|^2\right\rangle
=
\left\|\langle\mathbf{x}\rangle\right\|^2+\sigma_x^2
\end{equation}
therefore gives
\begin{equation}
\begin{aligned}
\sigma_{spatial}^2
&=
\frac{dk_BT}
{k_x+\frac{1}{2}k_LL^2},\\
\sigma_{spatial}^2
&\propto
\frac{1}
{k_x+\frac{1}{2}k_LL^2}.
\end{aligned}
\end{equation}

In scenarios where the label of a button is the primary driver of intent, the parameter $L$ outweighs the baseline attraction, $(k_L)L^2 \gg k_x$.
$$\sigma_{spatial} \approx \frac{1}{\sqrt{\frac{1}{2} (k_L) L^2}} = \frac{1}{L} \sqrt{\frac{2}{k_L}},$$
$$\therefore \sigma_{spatial} \propto \frac{1}{L}$$
\end{proof}

\subsubsection{Proof for \cref{theorem:non_equilibrium_theorem}} \label{subsubsection:non_equilibrium_theorem}
\begin{proof}
From \cref{definition:equipartition} \footnote{Note that we have used $T$ to denote kinetic energy in this proof rather than $K$.}:
\begin{equation}
\begin{aligned}
\langle T\rangle &= \langle V\rangle \\
\langle T\rangle-\langle V\rangle &= 0.
\end{aligned}
\end{equation}

Consider a cursor at rest near the edge of a large target. Its kinetic energy is zero, while the model assigns non-zero potential energy because the cursor is at a distance from the geometric center, from \cref{definition:mechanical_equilibrium}. Therefore,
\begin{equation}
\langle V\rangle-\langle T\rangle>0.
\end{equation}

Conversely, during extremely rapid cursor movement toward a target, such as in a shooting game, the kinetic energy may exceed the potential energy:
\begin{equation}
\langle T\rangle-\langle V\rangle>0.
\end{equation}

These inequalities can be combined as
\begin{equation}
\left|\langle T\rangle-\langle V\rangle\right|>0.
\end{equation}

Since the Lagrangian is defined as
\begin{equation}
\mathcal{L}=T-V,
\end{equation}
the same condition can be expressed as
\begin{equation}
\langle\mathcal{L}\rangle
=
\langle T\rangle-\langle V\rangle
\ne 0.
\end{equation}

Therefore, a non-zero averaged Lagrangian indicates that the interaction no longer satisfies the equilibrium condition of the harmonic agent-target model.
\end{proof}

\subsection*{Reproducibility Statement}

The target acquisition model is open source and available as the \texttt{intent-link} npmjs package at \href{https://www.npmjs.com/package/intent-link}{npmjs.com/package/intent-link}.

\section*{Generative AI Disclosure} \label{sec:ai_usage}
Claude Code (Model: Sonnet) was used through Visual Studio Code for code completion during the development of Experiment~3. The experiment, including its conditions, procedure, measurements, and analysis, was designed by the authors. Claude Code assisted with implementing some of the user interface elements shown in \cref{fig:experimental_setup_exp3}. The generated code was reviewed, modified, and tested by the authors before the experiment was conducted. 

\cref{fig:teaser_potential_well,fig:microstate_visualization} were generated by ChatGPT (Model: GPT-5.6 Sol). To illustrate the arbitrary damped harmonic oscillator in \cref{fig:damped_oscillator}, Gemini (Model: Pro) was used to write LaTeX code for the graph.

Generative AI was not used to generate experimental data, statistical results, other figures, interpretations, or conclusions. No other part of the project used Generative AI.

\begin{figure*}[t!]
    \centering
    \includegraphics[width=0.32\linewidth]{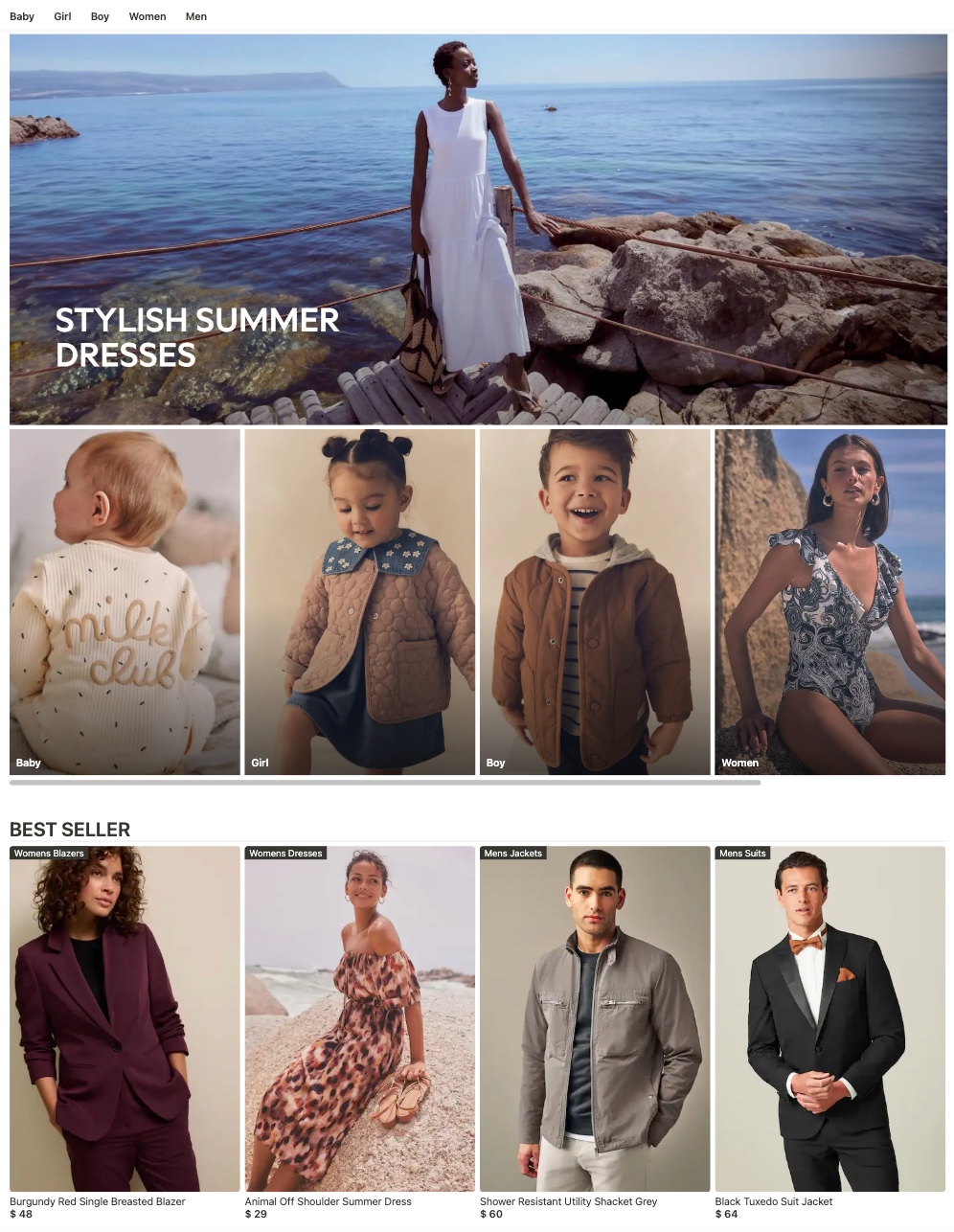}\hfill
    \includegraphics[width=0.32\linewidth]{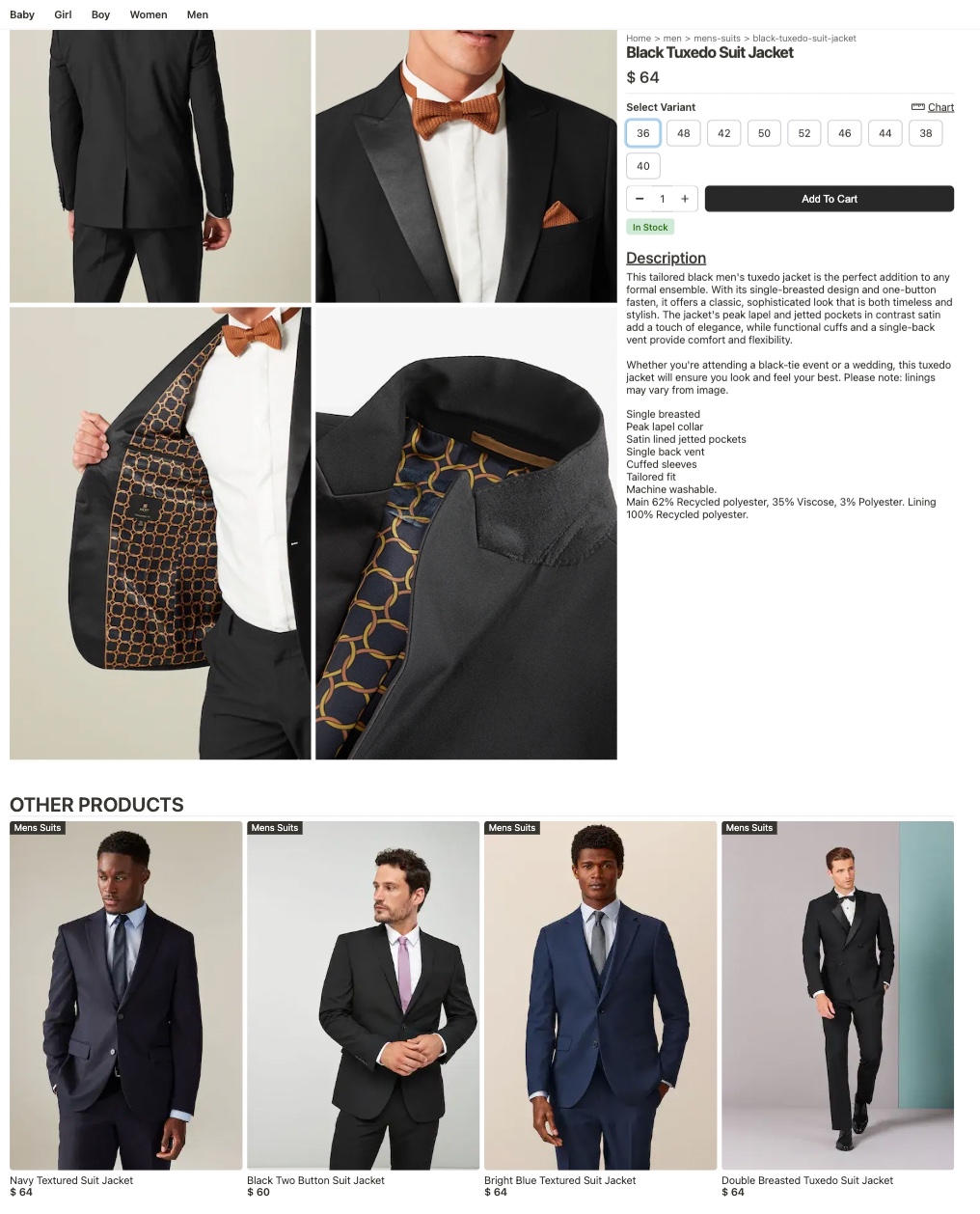}\hfill
    \includegraphics[width=0.32\linewidth]{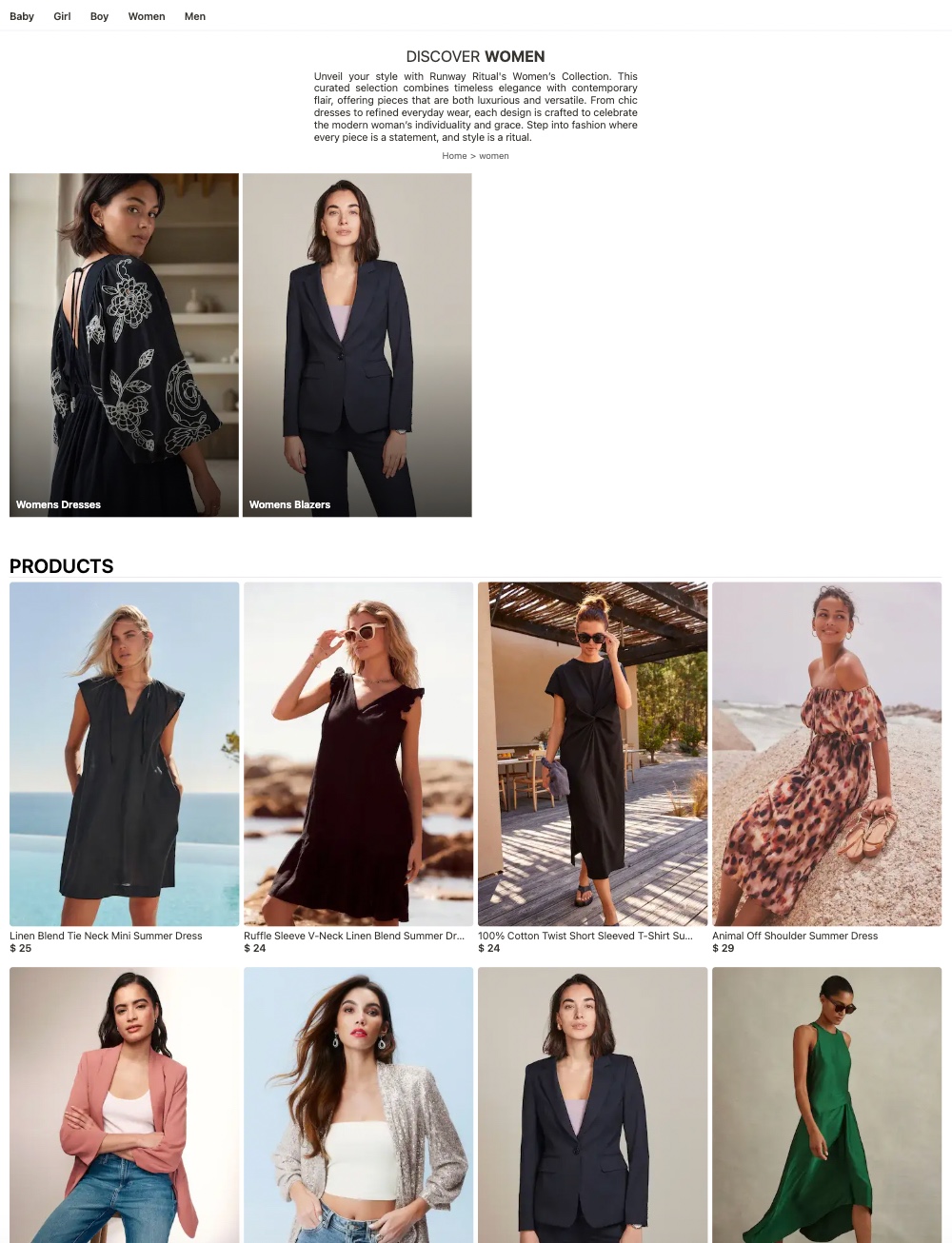}
    
    \caption{The custom e-commerce environment for Experiment 1 and 2. From left to right: The dense homepage, an individual product page, and a category browsing page.}
    \label{fig:experimental_setup_websites}
\end{figure*}

\begin{figure*}[t!]
    \centering
    \includegraphics[width=0.91\linewidth]{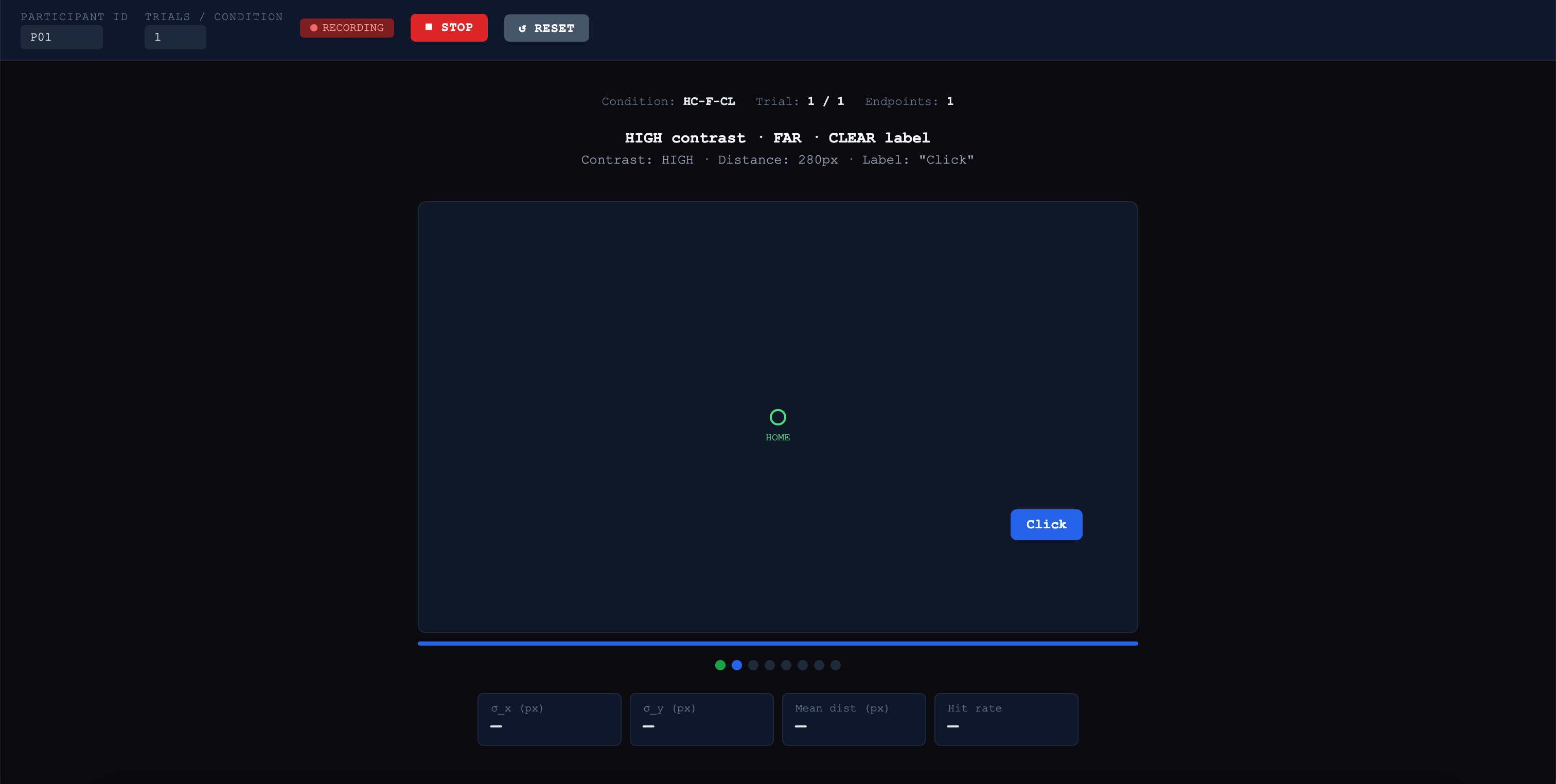}
    
    \vspace{0.1cm}
    
    \includegraphics[width=0.22\linewidth]{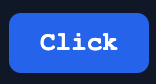}\hspace{0.1cm}
    \includegraphics[width=0.22\linewidth]{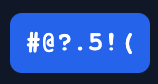}\hspace{0.1cm}
    \includegraphics[width=0.22\linewidth]{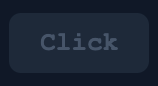}\hspace{0.1cm}
    \includegraphics[width=0.22\linewidth]{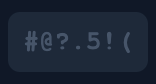}
    
    \caption{The controlled testing environment for Experiment 3. \textbf{Top:} The testing interface, which enforces a strict 280px target distance and records the standard deviation of the distribution of endpoints. \textbf{Bottom (Left to Right):} The 2$\times$2 target conditions: High Contrast/Clear Label, High Contrast/Ambiguous Label, Low Contrast/Clear Label, and Low Contrast/Ambiguous Label.}
    \label{fig:experimental_setup_exp3}
\end{figure*}

\end{document}